\documentclass[journal]{IEEEtran}
\usepackage{graphicx,amsmath,amssymb}
\usepackage{cite}
\usepackage{subfigure}
\usepackage{citesort}
\usepackage{fancyhdr}
\usepackage{mdwmath}
\usepackage{mdwtab}
\usepackage{balance}
\usepackage{xcolor}
\usepackage{bm}
\usepackage{amsthm}
\usepackage{algorithm}
\usepackage{algorithmic}
\usepackage{multirow}
\usepackage{flafter}
\usepackage{epstopdf}
\newtheorem{remark}{Remark}

\newtheorem{theorem}{Theorem}

\newtheorem{lemma}{Lemma}

\newtheorem{corollary}{Corollary}

\newtheorem{assumption}{Assumption}

\begin{document}
\title{Hardware Impaired Ambient Backscatter NOMA Systems: Reliability and Security}

\author{Xingwang~Li,~\IEEEmembership{Senior Member,~IEEE,}
        Mengle~Zhao,~\IEEEmembership{Student Member,~IEEE,}
        Ming~Zeng,~\IEEEmembership{Member,~IEEE,}
        Shahid~Mumtaz~\IEEEmembership{Senior Member,~IEEE,}
        Varun~G~Menon~\IEEEmembership{Senior Member,~IEEE,}
        Zhiguo~Ding,~\IEEEmembership{Fellow,~IEEE,}
        Octavia~A.~Dobre~\IEEEmembership{Fellow,~IEEE}

\thanks{X. Li and M. Zhao are with the School of Physics and Electronic Information Engineering, Henan Polytechnic University, Jiaozuo, China (email:lixingwangbupt@gmail.com, zhaomenglehpu@163.com) (First corresponding author: Xingwang  Li).}
\thanks{M. Zeng is with Memorial University, St.John's, NL A1B 3X9, Canada (email: mzeng@mun.ca).}
\thanks{S. Mumtaz is with Institute of Telecommunications, Aveiro, Portugal, (email: smumtaz@av.it.pt).}
\thanks{V. G. Menon is with Department of Computer Science and Engineering, SCMS School of Engineering and Technology, India. (email: varunmenon@ieee.org).}
\thanks{Z. Ding is with the School of Electrical and Electronic Engineering, The University of Manchester, Manchester, UK (email: zhiguo.ding@manchester.ac.uk).}
\thanks{O. A. Dobre is with Memorial University, St. John's, NL A1B 3X9, Canada (e-mail: odobre@mun.ca). }
}

\maketitle
\nocite{Li2018}
\begin{abstract}  Non-orthogonal multiple access (NOMA) and ambient backscatter communication have been envisioned as two promising technologies for the Internet-of-things due to their high spectral efficiency and energy efficiency. Motivated by this fact, we consider an ambient backscatter NOMA system in the presence of a malicious eavesdropper. Under some realistic assumptions of residual hardware impairments (RHIs), channel estimation errors (CEEs) and imperfect successive interference cancellation (ipSIC), we investigate the physical layer security (PLS) of the ambient backscatter NOMA systems focusing on  reliability and security. In order to further improve the security of the considered system, an artificial noise scheme is proposed where the radio frequency (RF) source acts as a jammer that transmits interference signal to the legitimate receivers and eavesdropper. On this basis, the analytical expressions for the outage probability (OP) and the intercept probability (IP) are derived. To gain more insights, the asymptotic analysis and diversity orders for the OP in the high signal-to-noise ratio (SNR) regime are carried out, and the asymptotic behaviors of the IP in the high main-to-eavesdropper ratio (MER) region are explored as well. Numerical results show that: 1) RHIs, CEEs and ipSIC have negative effects on the OP but positive effects on the IP; 2) Compared with CEEs, RHIs have a more serious impact on the reliability and security of the considered system; 3) There exists a trade-off between reliability and security, and this trade-off can be optimized by reducing the power coefficient of the artificial noise or increasing the interfering factor of readers; 4) There are error floors for the OP due to the CEEs and the reflection coefficient; 5) As MER grows large, the security for $R_n$ and $R_f$ is improved, while the security for $T$ is reduced.
\end{abstract}


\section{Introduction}
A large number of intelligent devices will be supported for the wireless networks with Internet-of-things (IoT) and massive machine-type communication  \cite{1,111}. To this end, non-orthogonal multiple access (NOMA) has been identified as a promising solution to serve massive connections due to high spectral efficiency and low latency \cite{2}.\footnote{Generally, NOMA can be divided into code-domain NOMA and power-NOMA. In this paper, we use NOMA to refer to the power-domain NOMA.} The distinguishing feature of NOMA is that a plurality of users are allowed to occupy the same time/frequency/code resources by power multiplexing through superposition coding \cite{3}. At the receiver, the signals can be extracted with the aid of successive interference cancellation (SIC) \cite{4}. From the perspective of coverage, NOMA can enhance the performance of the cell edge users by allocating more power to them \cite{7527668}.

On a parallel avenue, backscatter communication has emerged as a promising paradigm for green sustainable IoT applications due to its ultralow-power and low cost \cite{5}. A well-known backscatter communication application for the IoT is radio frequency identification (RFID) that consists of one reader and one tag. More exactly, the tag modulates and reflects the incident signal from the energy source through a mismatched antenna impedance to passively transmit information, and the reader performs demodulation after receiving the reflected signal \cite{6}. However, the traditional backscatter communication technology is limited by the power consumption resulting from the active transmission \cite{7}. To tackle this limitation, the work in \cite{8} proposed ambient backscatter prototypes. This technology utilizes environmental wireless signals (e.g., digital TV broadcasting or cellular signals) to collect energy and transmit information through battery-free tags.

Ambient backscatter technology has drawn great attention from both academia and industry \cite{9,10,11,12,13}. A framework for evaluating the  ultimate achievable rates of point-to-point networks with ambient backscatter devices was proposed in \cite{9}, where the impact of the backscatter transmission on the performance of the legacy systems was considered. In \cite{10}, the authors analyzed the outage performance of the ambient backscatter communication systems with a pair of passive tag-reader by deriving the exact and asymptotic expressions for the outage probability (OP). Guo \emph{et al}. in \cite{11} exploited the NOMA technology to support massive tag connections. According to the unique characteristics of the cooperative ambient backscatter system, the authors of \cite{12} proposed three symbiotic transmission schemes, where the relationships between the primary and backscatter transmissions were commensal, parasitic, and competitive. The authors in \cite{13} investigated the effects of co-channel interference and the energy harvesting (EH) on the achievable OP of the ambient backscatter communication systems with multiple backscatter links.

Another well-known fact is that the transmission of wireless signals is vulnerable to fronted threats due to the broadcast nature of wireless medium. The traditional key encryption technologies has high computation complexity, and thus, are not suitable for small-volume backscatter devices with limited storage and computing power \cite{14}. As a result, they may not be applied for solving the security communication problem of the ambient backscatter NOMA systems \cite{144}.

As an alternative, physical layer security (PLS) was proposed as a promising mechanism to ensure the security of wireless communication systems from an information theoretic perspective \cite{6772207,9013326}. By exploiting the inherent random characteristics of wireless channels, PLS can achieve secure communication for wireless networks without being eavesdropped by illegal eavesdroppers, which has sparked a great deal of research interests, e.g., see \cite{15,16,17,18,19} and the references therein. In \cite{15}, the secrecy outage performance of a multiple-relay NOMA network was investigated, where three relay selection schemes were proposed. With the emphasis on the cognitive radio networks (CRNs), the authors of \cite{16} evaluated the reliability-security tradeoff by deriving the connection outage probability and the secrecy outage probability for the cooperative NOMA aided CRN systems. Additionally, the secrecy rate was studied under the traditional backscatter communications systems in \cite{17}, where the reader and eavesdropper were equipped with multiple antennas. To enhance the security of the ambient backscatter communication systems, an optimal tag selection scheme for the multi-tag ambient backscatter systems was designed in \cite{18}. By the virtue of artificial noise, an enhanced PLS scheme for multi-tag ambient backscatter system was designed, and the bit error rate and secrecy rate were investigated in \cite{19}.  Moreover, the authors of \cite{8649584} proposed to combined multiple-input multiple-output technology with artificial noise technology to enhance the secrecy performance of NOMA systems.

Unfortunately, the common feature of the aforementioned contributions is that perfect radio frequency (RF) components are assumed, which may not be realistic in practical communication systems. In practice, all RF front-ends are vulnerable to several types of hardware impairments due to the configuration of low cost, low-power dissipation, and small size components, such as amplifier non-linearities, in-phase/quadrature imbalance, phase noise, and quantization error \cite{20,22,21}. These impairments can be generally eliminated by using some compensation and calibration algorithms. However, owing to some factors such as estimation errors, inaccurate calibration, and time-varying hardware characteristics, there are still some residual hardware impairments (RHIs), which can be modeled as an additive distortion noise to the transmitted/received signals \cite{20}. To this end, a great deal of works have studied the impact of RHIs on system performance \cite{21,23}. In \cite{21}, the authors investigated the effects of RHIs on the achievable sum rate of the unmanned aerial vehicle-aided NOMA relaying networks. Considering two types of relay selection schemes, the impacts of RHIs on the multiple-relay amplify-and-forward network was studied by deriving the tight closed-form expressions for the OP \cite{23}.

Moreover, another limitation of the above research works is that imperfect channel state information (CSI) is assumed available at receivers, which is not practical. In fact, it is a great challenge to obtain perfect channel knowledge due to channel estimation errors (CEEs) and feedback delay \cite{24}. The related research works about imperfect CSI have been reported in \cite{25,26,27}. The outage performance of the downlink cooperative NOMA systems based on wireless backhaul unreliability and imperfect CSI was studied by deriving the exact and asymptotic OP expressions at the receivers \cite{25}. A proportional fair scheduling algorithm was proposed to achieved high throughput and fairness, which was extended to the multi-user NOMA scenarios with imperfect CSI in \cite{26}. The authors of \cite{27} considered a more practical scenarios, where the outage performance of the amplify-and-forward relay systems was analyzed in the presence of RHIs and CEEs. Therefore, it is of high practical relevance to look into the realistic scenario with imperfect CSI and RHIs.

\subsection{Motivation and Contribution}
The previous research works have laid a solid foundation for the analysis of NOMA, ambient backscatter and PLS. However the joint effects of RHIs, CEEs and imperfect SIC (ipSIC) on the secure performance of the ambient backscatter NOMA systems have not yet been well investigated. To fill this gap, this paper makes an in-depth study of the joint effects of the three non-ideal factors on the reliability and the security of the ambient backscatter NOMA systems. In order to improve the security, we consider an artificial noise scheme, where the RF source sends the signal and artificial noise simultaneously. This scheme is feasible since it is carried out without changing the original system framework \cite{8879726,7546865}. Specifically, the analytical expressions for the OP and the intercept probability (IP) are derived for the far reader, the near reader and the tag, respectively. To obtain more insights, the asymptotic behaviors for the OP in the high signal-to-noise ratio (SNR) regime and the asymptotic behaviors for the IP in the high main-to-eavesdropper ratio (MER) region are explored. The essential contributions of this paper are summarized as follows:
\begin{itemize}

  \item  We consider a novel secure framework for the ambient backscatter NOMA systems in the presence of RHIs, CEEs, and ipSIC. To improve secure performance, an artificial noise scheme is designed.
  \item  We derive the analytical expressions for the OP and the IP the far reader, the near reader and the tag to evaluate the reliability and the security. The results show that a smaller power coefficient of artificial noise or a larger interfering factor of readers can enhance the impact of artificial noise on balancing the trade-off reliability-security.
  \item   In order to obtain deeper insights, we carry out the asymptotic analysis for the OP in the high SNR region as well as the diversity orders. Moreover, the asymptotic behaviors of the IP in the high MER regime are explored by introducing the MER. The obtained results indicate that there are error floors for the OP due to the CEEs and the reflection coefficient.
\end{itemize}

\subsection{Organization and Notations}
The remainder of this paper is organised as follows. In Section II, we introduce the ambient backscatter NOMA model. In Section III, the reliability is investigated by deriving the analytical and asymptotic expressions for the OP, while the expressions of IP are derived to analyze the security. In Section IV, some numerical results are provided to validate the correctness of the theoretical analysis. Section V concludes the paper and summarizes key findings.

We use $E\left\{  \cdot  \right\}$ to denote the expectation operation. A complex Gaussian random variable with mean $\mu$ and variance $\sigma^2$ reads as $\mathcal{CN}\{\mu, \sigma^2\}$. $\Pr \left\{  {\cdot}  \right\}$ denotes the probability and ${\rm{K}}v\left(  \cdot  \right)$ represents the v-th order modified
Bessel function of the second kind, while $n!$ denotes the factorial operation. Finally, $f_X(\cdot)$ and $F_X(\cdot)$ are the probability density function (PDF) and the cumulative distribution function (CDF) of a random variable, respectively.

\begin{figure}[!t]
\setlength{\abovecaptionskip}{0pt}
\centering
\includegraphics [width=2.5in]{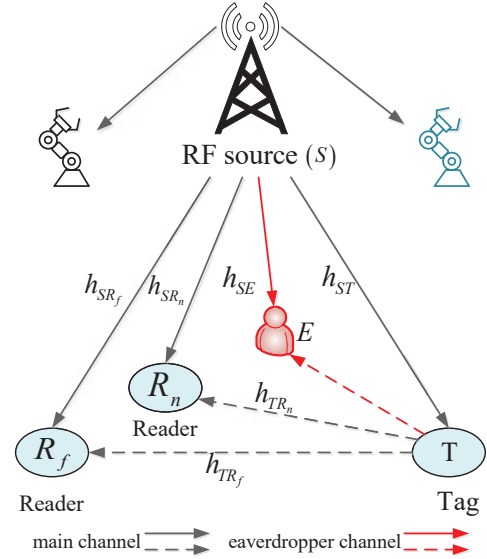}
\caption{Ambient backscatter NOMA system model.}
\label{fig2}
\end{figure}

\section{System Model}\label{sec2}
As illustrated in Fig. 1, we consider a downlink ambient backscatter NOMA system, which consists of one ambient RF source ($S$), one tag ($T$), two readers ($R_f$, $R_n$) and one eavesdropper ($E$). In this study, $S$ transmits the signal to readers and tag in the same resource block. Meanwhile, $T$ transmits its own information to the readers by reflecting the signals from $S$ signal, whereas $E$ can intercept the signal intended for readers. We consider the following assumptions: i) All the nodes are equipped with a single antenna; ii) RHIs exist at $S$, readers and $E$ but not at the tag; iii) All links $h$ are subject to Rayleigh fading.

Under practical considerations, the perfect CSI may be unavailable due to some CEEs. The common way to obtain CSI is channel estimation. For this purpose, by adopting linear minimum mean square error (MMSE), the channel can be modeled as ${h_{AB}} = {\hat h_{AB}} + {e_{AB}}$ \cite{28}, where ${{\hat h}_{AB}}$ is the estimated channel of ${h_{AB}}$, and ${e_{AB}}\sim {\rm{{\cal C}{\cal N}}}\left( {0,\sigma _{{e_{AB}}}^2} \right)$ denotes the corresponding channel estimation errors, where the variance of CEE $\sigma _{{e_{AB}}}^2$ indicates the quality of CSI.

To improve the security communication of ambient backscatter NOMA systems, we consider injecting artificial noise $z\left( t \right)$ with $E\left( {{{\left| {z\left( t \right)} \right|}^2}} \right) = 1$ at $S$. Due to the CEEs, the artificial noise will cause interference to the readers and eavesdropper. Then, the superposition message at $S$ can be written as
\begin{equation}
\label{1}
{x_s} = {\sqrt {{a_1}{P_s}} {x_1}} + {\sqrt {{a_2}{P_s}} {x_2}}+\sqrt {{P_J}} z\left( t \right),
\end{equation}
where ${P_S}$ is the transmit power for the desired signals at $S$; ${a_1}$ and ${a_2}$ are the power allocation coefficients for the near reader and the far reader with ${a_1} + {a_2} \!= \!1$ and ${a_1} \!<\! {a_2}$, respectively; ${x_1}$ and ${x_2}$ are the corresponding transmitted signals of $R_n$ and $R_f$ with $E\left( {{{\left| {{x_1}} \right|}^2}} \right)\! =\! E\left( {{{\left| {{x_2}} \right|}^2}} \right)\! =\! 1$; ${P_J}$ is the transmitted power of the artificial noise with ${P_J}=\varphi _J {P_S}$, with $\varphi _J  \in \left( {0,1} \right]$ as the power coefficient of artificial noise.

Next, $T$ backscatters the $S$ signal to $R_f$, $R_n$ and $E$ with its own signal $c\left( t \right)$, with $E\left( {{{\left| {c\left( t \right)} \right|}^2}} \right) = 1$. Therefore, $R_f$, $R_n$ and $E$ receive the signals from $S$ and the backscattered from $T$. Considering the RHIs and CEEs, the received signals at $i$ ($i \in \left( {{R_f},{R_n},E} \right)$) can be expressed as
\begin{equation}
\label{2}
{y_i} = \beta {h_{Ti}}{h_{ST}}\left( {{x_s}c\left( t \right){\rm{ + }}{\eta _{Si}}} \right){\rm{ + }}{h_{Si}}\left( {{x_s}{\rm{ + }}{\eta _{Si}}} \right){\rm{ + }}{n_i},
\end{equation}
where $\beta $ is a complex reflection coefficient used to normalize $c\left( t \right)$; ${n_i} \sim {\rm{{\cal C}{\cal N}}}\left( {0,{N_0}} \right)$ is the complex additive white Gaussian noise (AWGN); ${\eta _{Si}}\sim {\rm{{\cal C}{\cal N}}}\left( {0,\kappa _{Si}^2{P_S}} \right)$; $\kappa _{Si}$ denotes the level of hardware impairment at transceivers, which can be measured in practice based on the error vector magnitude (EVM) \cite{29}; ${h_{Si}}$, ${h_{Ti}}$ and ${h_{ST}}$ are the channel coefficients $S \to i$, $T \to i$ and $S \to T$, respectively.

According to the NOMA protocol, $R_f$ can decode the signals ${x_2}$, and $R_n$ and $E$ can decode the signals ${x_2}$, ${x_1}$ and $c\left( t \right)$ in turn with the aid of SIC. In addition, the readers can only eliminate part of the interference due to the presence of CEEs. Then, the received signal-to-interference-plus-noise ratio (SINR) of $i$ $\left( {i \in \{ {R_n},{R_f},E\} } \right)$ can be given as\footnote{It should be pointed out that $R_f$ only needs to decode its own signal $x_2$, that is, the SINR of $R_f$ is $\gamma _{{R_f}}^{{x_2}}$.}
\begin{equation}
\label{3}
\gamma _i^{{x_2}} \!\!=\!\! \frac{{{{\left| {{{\hat h}_{Si}}}\right|}^2}{a_2}\gamma }}{{\gamma \left[ \!{{{\left| {{{\hat h}_{ST}}} \!\right|}^2}\left(\! {{B_i}{{\left| \!{{{\hat h}_{Ti}}} \!\right|}^2} \!\!+\! {M_i}} \right) \!\!+\! {C_i}{{\left| \!{{{\hat h}_{Ti}}} \!\right|}^2} \!\!+\! {Q_i}{{\left| \!{{{\hat h}_{Si}}} \!\right|}^2} \!\!+\! {\psi _i}} \right] \!\!+\! 1}},
\end{equation}
\begin{equation}
\label{4}
\gamma _i^{{x_1}} \!=\! \frac{{{{\left| {{{\hat h}_{Si}}} \right|}^2}{a_1}\gamma }}{{\gamma \!\left[\! {{{\left| \!{{{\hat h}_{ST}}} \!\right|}^2}\left( \!{{B_i}{{\left| {{{\hat h}_{Ti}}} \!\right|}^2} \!\!+\! {M_i}} \!\right) \!\!+\! {C_i}{{\left| \!{{{\hat h}_{Ti}}} \!\right|}^2} \!\!+\! {O_i}{{\left| {{{\hat h}_{Si}}} \right|}^2} \!\!+\! {\psi _i}} \!\right] \!\!+\! 1}},
\end{equation}
\begin{equation}
\label{5}
\gamma _i^{c\left( t \right)}\! =\! \frac{{{\beta ^2}{{\left| {{{\hat h}_{Ti}}} \right|}^2}{{\left| {{{\hat h}_{ST}}} \right|}^2}\gamma }}{{\gamma\! \left[\! {{{\left|\! {{{\hat h}_{ST}}} \!\right|}^2}\left( \!{{m_i}{{\left| \!{{{\hat h}_{Ti}}}\! \right|}^2} \!\!+\! {M_i}} \!\right) \!\!+\! {C_i}{{\left|\! {{{\hat h}_{Ti}}} \!\right|}^2} \!\!+\! {\xi _i}{{\left| \!{{{\hat h}_{Si}}} \!\right|}^2} \!+\! {\psi_i}} \!\right] \!\!+\! 1}},
\end{equation}
where \!$\gamma  = {P_S}/{N_0}$ \!represents the transmit SNR at $S$; $\varepsilon $ is the parameter of ipSIC;
${B_{{R_f}}} \!\!\!= \!\!{\beta ^2}\!\left(\! {1 \!\!+\! \varpi \!{\varphi _J} \!\!+\! \kappa _{S{R_f}}^2}\!\! \right)$, ${C_{{R_f}}} \!\!\!=\!\! {B_{{R_f}}}\!\sigma _{{e_{ST}}}^2$, ${M_{{R_f}}} = {B_{{R_f}}}\sigma _{{e_{T{R_f}}}}^2$, ${Q_{{R_f}}} = \gamma \left( {{a_1} + \varpi {\varphi _J} + \kappa _{S{R_f}}^2} \right)$, ${\psi _{{R_f}}} = {B_{{R_f}}}\sigma _{{e_{T{R_f}}}}^2\sigma _{{e_{ST}}}^2 + \sigma _{{e_{S{R_f}}}}^2\left( {1 + \varpi {\varphi _J} + \kappa _{S{R_f}}^2} \right)$;
$\varpi$ is the interference factor, reflecting the degree of interference of the artificial noise to the readers, with $0 \;\le\; \varpi \;\le\; 1$;
${B_{{R_n}}} \!=\! {\beta ^2}\left( {1 \!+\! \varpi {\varphi _J} \!+\! \kappa _{S{R_n}}^2} \right)$, ${C_{{R_n}}} \!=\! {B_{{R_n}}}\sigma _{{e_{ST}}}^2$, ${M_{{R_n}}} \!= \!{B_{{R_n}}}\sigma _{{e_{T{R_n}}}}^2$, ${Q_{{R_n}}} = \gamma \left( {{a_1} + \varpi {\varphi _J} + \kappa _{S{R_n}}^2} \right)$, ${\psi _{{R_n}}} = {B_{{R_n}}}\sigma _{{e_{T{R_n}}}}^2\sigma _{{e_{ST}}}^2 + \sigma _{{e_{S{R_n}}}}^2\left( {1 + \varpi {\varphi _J} + \kappa _{S{R_n}}^2} \right)$, ${O_{{R_n}}} = \left( {\varepsilon {a_2} + \varpi {\varphi _J} + \kappa _{S{R_n}}^2} \right)$, ${m_{{R_n}}} = {\beta ^2}\left( {\kappa _{S{R_n}}^2 + \varpi {\varphi _J}} \right)$, ${\xi _{{R_n}}} = \left( {\varepsilon  + \varpi {\varphi _J} + \kappa _{S{R_n}}^2} \right)$;
${B_{{E}}} \!=\! {\beta ^2}\left( {1 \!+\!  {\varphi _J} \!+\! \kappa _{S{E}}^2} \right)$, ${C_{{E}}} \!=\! {B_{{E}}}\sigma _{{e_{ST}}}^2$, ${M_{{E}}} \!= \!{B_{{E}}}\sigma _{{e_{T{E}}}}^2$, ${Q_{{E}}} = \gamma \left( {{a_1} +  {\varphi _J} + \kappa _{S{E}}^2} \right)$, ${\psi _{{E}}} = {B_{{E}}}\sigma _{{e_{T{E}}}}^2\sigma _{{e_{ST}}}^2 + \sigma _{{e_{S{E}}}}^2\left( {1 +  {\varphi _J} + \kappa _{S{E}}^2} \right)$, ${O_{{E}}} = \left( {\varepsilon {a_2} + {\varphi _J} + \kappa _{S{E}}^2} \right)$, ${m_{{E}}} = {\beta ^2}\left( {\kappa _{S{E}}^2 +  {\varphi _J}} \right)$, ${\xi _{{E}}} = \left( {\varepsilon  +  {\varphi _J} + \kappa _{S{E}}^2} \right)$.

\section{Performance Analysis}\label{sec3}
In this section, we investigate the reliability and security of the ambient backscatter NOMA systems in term of OP and IP. In addition, the asymptotic OP and diversity orders in the high SNR regions are examined, as well as the the asymptotic IP in the high MER regime.
\subsection{OP Analysis}

\emph{1) OP for $R_f$}

The outage event occurs at $R_f$ when $R_f$ cannot successfully decode ${x_2}$. Thus, the OP at $R_f$ can be expressed as
\begin{equation}
\label{6}
P_{out}^{{R_f}} = 1 - {{\rm{P}}_r}\left( {\gamma _{{R_f}}^{{x_2}} > {\gamma _{th2}^{{R_f}}}} \right),
\end{equation}
where ${\gamma _{th2}^{{R_f}}}$ is the target rate of $R_f$.

\begin{theorem}
For Rayleigh fading channels, the analytical expression for the OP of the far reader can be obtained as
\begin{equation}
\label{7}
P_{out}^{{R_f}} \!=\! 1 \!+\! {\Delta _2^{{R_f}}}{e^{{\Delta _1^{{R_f}}} \!-\! {\Delta _3^{{R_f}}} \!-\! \frac{{{\gamma _{th2}^{{R_f}}}}}{{{\lambda _{S{R_f}}}\gamma \left( {{a_2} \!-\! {Q_{{R_f}}}{\gamma _{th2}^{{R_f}}}} \right)}}}}{\rm{Ei}}\left( { \!-\! {\Delta _1^{{R_f}}}} \right),
\end{equation}
\end{theorem}
\noindent where $\Delta _1^i \!\!=\! \left( {\frac{{{M_i}{\gamma _{th2}^{{i}}}}}{{{\lambda _{Si}}\left( {{a_2} \!-\! {Q_i}\gamma _{th2}^i} \right)}} \!+\! \frac{1}{{{\lambda _{ST}}}}} \right)\frac{{{\lambda _{Si}}\left( {{a_2} \!-\! {Q_i}\gamma _{th2}^i} \right) \!+\! {\lambda _{Ti}}{C_i}\gamma _{th2}^i}}{{{\lambda _{Ti}}{B_i}\gamma _{th2}^i}}$, $\Delta _2^{i} \!=\! \frac{{{\lambda _{Si}}\left(\! {{a_2} \!-\! {Q_i}\gamma _{th2}^i} \!\right)}}{{{\lambda _{ST}}{\lambda _{Ti}}{B_i}\gamma _{th2}^i}}$, $\Delta _3^i\! =\! \frac{{{\psi_i}{\gamma _{th2}^i}}}{{{\lambda _{Si}}\left(\! {{a_2} \!-\! {Q_i}\gamma _{th2}^i} \!\right)}} $, (${i \in \{ {R_n},{R_f},E\} } $).
${\rm{Ei}}\left( p \right)$ is the exponential integral function \cite{31} and is expressed by
\begin{equation}
\label{8}
{\rm{Ei}}\left( p \right) \!=\! \frac{{{{\left( { \!-\! p} \right)}^{i \!-\! 1}}}}{{\left( {i \!-\! 1} \right)!}}\left[ { - \ln p \!+\! \psi \left( i \right)} \right] \!-\! \sum\limits_{m = 0}^\infty  {\frac{{{{( \!-\! p)}^m}}}{{\left( {m \!-\! i \!+\! 1} \right)m!}}} ,
\end{equation}
with
\begin{equation}
\label{9}
\left\{ \begin{array}{l}
\psi \left( 1 \right) =  - \upsilon \\
\psi \left( i \right) =  - \upsilon  + \sum\limits_{m = 1}^{i - 1} {\frac{1}{m}{\rm{    }} \;\;\;\;i > 1}
\end{array} \right. ,
\end{equation}
where $\nu  \approx 0.577$ is the Euler constant.
\begin{proof}
See Appendix A.
\end{proof}

\begin{corollary}
At high SNRs, the asymptotic expression for the OP of $R_f$ of the ambient backscatter NOMA systems is given as
\begin{equation}
\label{10}
P_{out,\infty }^{{R_f }} = 1 + {\Delta _2^{{R_f}}}{e^{{\Delta _1^{{R_f}}} - {\Delta _3^{{R_f}}}}}{\rm{Ei}}\left( { - {\Delta _1^{{R_f}}}} \right).
\end{equation}
\end{corollary}

\emph{2) OP for $R_n$}

To successfully decode ${x_1}$ at $R_n$, two conditions are needed to be met simultaneously: 1) $R_n$ can successfully decode ${x_2}$; 2) $R_n$ can successfully decode its own information ${x_1}$. Therefore, the OP of $R_n$ can be expressed as
\begin{equation}
\label{11}
P_{out}^{{R_n}} = 1 - {{\rm{P}}_r}\left( {\gamma _{{R_n}}^{{x_2}} > {\gamma _{th2}^{{R_n}}},\gamma _{{R_n}}^{{x_1}} > {\gamma _{th1}^{{R_n}}}} \right),
\end{equation}
where ${\gamma _{th1}^{{R_n}}}$ is the target rate of $R_n$.
\begin{theorem}
For Rayleigh fading channels, the analytical expression for the OP of the near reader can be obtained as
\begin{equation}
\label{15}
P_{out}^{{R_n}} \!=\!\! 1 \!+\! \frac{{{\lambda _{S{R_n}}}}}{{{\lambda _{ST}}{\varsigma _{{R_n}}} {\lambda _{T{R_n}}}{B_{{R_n}}}}}\!{e^{ \!-\! \left( {\frac{{\varsigma _{{R_n}}} }{{{\lambda _{S{R_n}}}\gamma }}  \!+\! {\Delta _4^{{R_n}}}} \right)}}{\rm{Ei}}\!\left(\! { \!-\! {\Delta _4^{{R_n}}}}\! \right),
\end{equation}
where ${\varsigma _i} = \max \left\{ {\frac{{\gamma _{th1}^i}}{{{a_1} - {O_i}\gamma _{th1}^i}},\frac{{\gamma _{th2}^i}}{{{a_2} - {Q_i}\gamma _{th2}^i}}} \right\}$, $\Delta _4^i = \frac{{\left( {{\lambda _{ST}}{\varsigma _i}{M_i} + {\lambda _{S{i}}}} \right)\left( {{\lambda _{Si}} + {\varsigma _i}{\lambda _{Ti}}{C_i}} \right)}}{{{\lambda _{Si}}{\lambda _{ST}}{\varsigma _i}{\lambda _{Ti}}{B_i}}}+ \frac{{{\varsigma _i} {\psi_{i}}}}{{{\lambda _{S{i}}}}}$, (${i \in \{ {R_n},{R_f},E\} } $).
\end{theorem}

\begin{proof}
By substituting (3) and (4) into (11), we can obtain the result of (12) after some mathematical manipulations, as in the proof of \textbf{Theorem 1}.
\end{proof}

\begin{corollary}
At high SNRs, the asymptotic expression for the OP of $R_n$ of the ambient backscatter NOMA systems is given as
\begin{equation}
\label{16}
P_{out,\infty }^{{R_n}} \!\!=\!\! 1 \!+\! \frac{{{\lambda _{S{R_n}}}}}{{{\lambda _{ST}}{\varsigma _{{R_n}}} {\lambda _{T{R_n}}}{B_{{R_n}}}}}{e^{ - \left( {\frac{{{\varsigma _{{R_n}}} {\psi_{{R_n}}}}}{{{\lambda _{S{R_n}}}}} \!+\! {\Delta _4^{{R_n}}}} \right)}}{\rm{Ei}}\left( { \!-\! {\Delta _4^{{R_n}}}} \right).
\end{equation}
\end{corollary}

\emph{3) OP for $T$}

The T signals can be successfully decoded when ${x_2}$ and ${x_1}$ are perfectly decoded at $R_n$. Thus, the OP of BD can be expressed as
\begin{equation}
\label{14}
P_{out}^{T} \!=\! 1 \!-\! {{\rm{P}}_r}\left( {\gamma _{{R_n}}^{{x_2}} > {\gamma _{th2}^{{R_n}}},\gamma _{{R_n}}^{{x_1}} > {\gamma _{th1}^{{R_n}}},\gamma _{{R_n}}^{c\left( t \right)} > {\gamma _{thc}^{{R_n}}}} \right),
\end{equation}
where ${\gamma _{thc}}$ is the target rate for $R_n$ decoding tag signals.
\begin{theorem}
For Rayleigh fading channels, we have

$\bullet$ Non-ideal conditions

The analytical expression for the OP of T in (15) is provided at the top of next page.
\begin{figure*}[!t]\label{18}
\normalsize
\begin{align}\nonumber
\label{15}
&P_{out}^{T,ni} = 1 - \frac{{2{\lambda _{S{R_n}}}}}{{{\lambda _{T{R_n}}}{\lambda _{ST}}{\varsigma _{{R_n}}}{B_{{R_n}}}}}{e^{ - \left( {{B_5} + \frac{{{\lambda _{T{R_n}}}{\varsigma _{{R_n}}}{B_{{R_n}}}{\gamma _{thc}}}}{{{\lambda _{S{R_n}}}{\lambda _{T{R_n}}}\gamma {\Delta _5}}} + \frac{{{\varsigma _{{R_n}}}}}{{{\lambda _{S{R_n}}}\gamma }}} \right)}}\sum\limits_{v = 1}^\infty  {{{\left( { - 1} \right)}^v}\frac{1}{{B_4^v}}} {\left( {\frac{{\left( {{B_3} + {\Delta _6}} \right)}}{{{B_1}}}} \right)^{\frac{v}{2}}}{K_v}\left( {2\sqrt {\left( {{B_3} + {\Delta _6}} \right){B_1}} } \right) + \\
&\frac{{{\lambda _{S{R_n}}}{\xi_{{R_n}}}\gamma _{thc}^{{R_n}}}}{{{\lambda _{ST}}{\lambda _{T{R_n}}}\Delta _5^{{R_n}}}}{e^{A_2^{{R_n}}}}\!\!\!\left(\! {\frac{\pi }{N}\!\sum\limits_{k = 0}^N \!{{e^{ \!-\! \left(\! {\frac{{2\left(\! {A_3^{{R_n}} \!+\! \Delta _8^{{R_n}}} \!\right)}}{{A_4^{{R_n}}\left( {{\vartheta _k} \!+ \!1} \right)}} - \frac{{A_1^{{R_n}}A_4^{{R_n}}\left( \!{{\vartheta _k} \!+\! 1} \!\right)}}{2}} \!\right)}}\!\!\sqrt {1 \!-\! \vartheta _k^2} \left(\! {\frac{1}{{{\vartheta _k} \!+\! 3}} \!-\! \frac{1}{{{\vartheta _k}\! +\! 1}}} \right) \!+\! 2{K_0}\left(\! {2\sqrt { \!- \!A_1^{{R_n}}\left( \!{A_3^{{R_n}} \!+\! \Delta _8^{{R_n}}} \!\right)} } \!\right)} } \!\right).
\end{align}
\hrulefill \vspace*{0pt}
\end{figure*}

\noindent{In (15), ${\vartheta _k} = \cos \left[ {\left( {2k - 1} \right)\pi /\left( {2N} \right)} \right]$, N is an accuracy-complexity trade-off parameter. $\Delta _5^i = {\beta ^2} - {m_{i}}\gamma _{thc}^i $, ${\Delta _6} = \left[\! {{\lambda _{\!T\!{R_n}}}{\varsigma _{{R_n}}}{C_{{\!R_n}}}\gamma _{\!thc}^{{\!R_n}}\left( \!{{B_{{\!R_n}}}\gamma \gamma _{thc}^{{R_n}} \!\!+\!\! {\Delta _5}}\! \right) \!+\! {\lambda _{\!S{R_n}}}\gamma _{thc}^{{R_n}}\Delta _5^{{\!R_n}}} \!\right]\!/\!\left( \!{\Delta _5^{{R_n}}\gamma } \!\right)$, $\Delta _7^i = \left( {{\lambda _{Si}}{\xi_i} - {\lambda _{Ti}}{C_i}} \right)\gamma _{thc}^i$, $\Delta _8^i = \left( {{\lambda _{Ti}}{C_i}\gamma _{thc}^i + \Delta _7^i} \right)\gamma $,
$A_1^i \!\!=\!\!  \frac{-1}{{{\lambda _{Si}}{\xi_i}{\lambda _{ST}}{\lambda _{Ti}}\Delta _5^i{\gamma ^2}}}$,
$A_2^i =  - \left( {\frac{{{C_i}\gamma _{thc}^i}}{{{\lambda _{ST}}\Delta _5^i}} + \frac{{{M_i}\gamma _{thc}^i}}{{{\lambda _{Ti}}\Delta _5^i}}} \right)$,
$A_3^i = \frac{{{\lambda _{Ti}}{M_i}{{\left( {{C_i}\gamma \gamma _{thc}^i} \right)}^2} +\; \left( {{\lambda _{Ti}}{\psi_i}\Delta _5^i + \Delta _7^i{M_i}} \right){C_i}{\gamma ^2}\gamma _{thc}^i}}{{\Delta _5^i}} \;+ \Delta _7^i{\psi_i}{\gamma ^2}$,\\
$A_4^i \!=\! \lambda _{Si}^2\xi_i^2\gamma _{thc}^i$,
${B_1} \!\!=\!\! \frac{{{\lambda _{S{R_n}}} \!+\! {\lambda _{ST}}{\varsigma _{{R_n}}}{M_{{R_n}}}}}{{\lambda _{S{R_n}}^2{\lambda _{ST}}{\lambda _{T{R_n}}}\gamma \Delta _5^{{R_n}}}} \!+\! \frac{{{\varsigma _{{R_n}}}{C_{{R_n}}}{B_{{R_n}}}\gamma _{thc}^{{R_n}}}}{{\lambda _{S{R_n}}^2{\lambda _{T{R_n}}}{{\left( {\Delta _5^{{R_n}}} \right)}^2}}}$, (${i \in \{ {R_n},{R_f},E\} } $),
${B_5} = \frac{{\left( {{\lambda _{ST}}{\varsigma _{{R_n}}}{C_{{R_n}}} + {\lambda _{S{R_n}}}} \right){C_{{R_n}}}\gamma _{thc}^{{R_n}}}}{{{\lambda _{ST}}{\lambda _{S{R_n}}}\Delta _5^{{R_n}}}} + {B_2} + \frac{{{\varsigma _{{R_n}}}{\psi_{{R_n}}}}}{{{\lambda _{S{R_n}}}}}$, where ${B_2}$, ${B_3}$ and ${B_4}$ are provided at the top of the next page.}
\begin{figure*}[!t]\label{18}
\normalsize
\begin{align}
\label{18}
{B_2} =& \frac{{2{\varsigma _{{R_n}}}{B_{{R_n}}}{M_{{R_n}}}{C_{{R_n}}}{{\left( {\gamma _{thc}^{{R_n}}} \right)}^2}}}{{{\lambda _{S{R_n}}}{{\left( {\Delta _5^{{R_n}}} \right)}^2}}}{\rm{ + }}\frac{{\left( {{\lambda _{B{D_n}}}{\varsigma _{{R_n}}}{C_{{R_n}}}{M_{{R_n}}} + {\lambda _{S{R_n}}}{M_{{R_n}}} + {\lambda _{B{R_n}}}{\varsigma _{{R_n}}}{B_{{R_n}}}{\psi _{{R_n}}}} \right)\gamma _{thc}^{{R_n}}}}{{{\lambda _{S{R_n}}}{\lambda _{B{R_n}}}\Delta _5^{{R_n}}}},\\ \nonumber
{B_3} =& \left[ {{\lambda _{T{R_n}}}{\varsigma _{{R_n}}}{M_{{R_n}}}C_{{R_n}}^2\gamma {{\left( {\gamma _{thc}^{{R_n}}} \right)}^2}\left( {{B_{{R_n}}}\gamma _{thc}^{{R_n}} + \Delta _5^{{R_n}}} \right) + \left( {{\lambda _{S{R_n}}}{M_{{R_n}}} + {\lambda _{T{R_n}}}{\varsigma _{{R_n}}}{B_{{R_n}}}{\psi_{{R_n}}}} \right){C_{{R_n}}}\gamma {{\left( {\Delta _5^{{R_n}}\gamma _{thc}^{{R_n}}} \right)}^2}} \right]/{\left( {\Delta _5^{{R_n}}} \right)^2} \\
&+ \left( {{\lambda _{T{R_n}}}{\varsigma _{{R_n}}}{C_{{R_n}}} + {\lambda _{S{R_n}}}} \right){\psi_{{R_n}}}\gamma \gamma _{thc}^{{R_n}},\\
{B_4} =& \left[ {{\varsigma _{{R_n}}}{\lambda _{S{R_n}}}{\lambda _{B{R_n}}}{C_{{R_n}}}\gamma \left( {{B_{{R_n}}}\gamma _{thc}^{{R_n}} + \Delta _5^{{R_n}}} \right) + \lambda _{S{R_n}}^2\Delta _5^{{R_n}}\gamma } \right]/\left( {{\varsigma _{{R_n}}}{B_{{R_n}}}} \right).
\end{align}
\hrulefill \vspace*{0pt}
\end{figure*}

$\bullet$ Ideal conditions

For ideal conditions, the analytical expression for the OP of the BD in (19) is provided at the top of next page.
\begin{figure*}[!t]\label{18}
\normalsize
\begin{align}\nonumber
\label{15}
P_{out}^{T,id} =& 1 + {\Delta _9}{e^{{\Delta _9} - \frac{{{\varsigma _{{R_n}}}}}{{{\lambda _{S{R_n}}}\gamma}}}}{\rm{Ei}}\left( { - {\Delta _9}} \right) + \frac{{\gamma _{thc}^{{R_n}}\pi }}{{N{\lambda _{T{R_n}}}{\lambda _{ST}}\gamma\Delta _5^{{R_n}}}}\sum\limits_{k = 0}^N {{e^{ - \left( {{\varsigma _{{R_n}}}{B_{{R_n}}}{\Delta _{10}}{\rm{ + }}\frac{{{\varsigma _{{R_n}}}}}{{{\lambda _{S{R_n}}}\gamma}}} \right)}}{K_0}\left( {2\sqrt {{\Delta _{10}}} } \right)} \sqrt {1 - \vartheta _k^2}   \\
&-{\Delta _{11}}{e^{{\Delta _{11}} + \frac{1}{{{\lambda _{S{R_n}}}\gamma{\xi_{{R_n}}}}}}}{\rm{Ei}}\left( { - {\Delta _{11}}} \right) - \frac{{\gamma _{thc}^{{R_n}}\pi }}{{N{\lambda _{T{R_n}}}{\lambda _{ST}}\gamma\Delta _5^{{R_n}}}}\sum\limits_{k = 0}^N {{e^{\frac{1}{{{\lambda _{S{R_n}}}\gamma{\xi_{{R_n}}}}} - \frac{{{\vartheta _k} + 1}}{{2{\lambda _{S{R_n}}}\gamma{\xi_{{R_n}}}}}}}{K_0}\left( {2\sqrt {{\Delta _{10}}} } \right)} \sqrt {1 - \vartheta _k^2}  .
\end{align}
\hrulefill \vspace*{0pt}
\end{figure*}

\noindent{In (19), ${\Delta _9} = \frac{{{\lambda _{S{R_n}}}}}{{{\lambda _{T{R_n}}}{\lambda _{ST}}{\varsigma _{{R_n}}}{B_{{R_n}}}}}$, ${\Delta _{10}} = \frac{{\left( {{\vartheta _k} + 1} \right)\gamma _{thc}^{{R_n}}}}{{2{\lambda _{T{R_n}}}{\lambda _{ST}}\gamma \Delta _5^{{R_n}}}}$, ${\Delta _{11}} = \frac{{{\lambda _{S{R_n}}}{\xi_{{R_n}}}\gamma _{thc}^{{R_n}}}}{{{\lambda _{T{R_n}}}{\lambda _{ST}}\Delta _5^{{R_n}}}}$, and ${\Delta _{12}} = \frac{{\left( {{\vartheta _k} + 1} \right)\gamma _{thc}^{{R_n}}}}{{2{\lambda _{T{R_n}}}{\lambda _{ST}}\gamma\Delta _5^{{R_n}}}}$.}

\end{theorem}

\begin{proof}
See Appendix B.
\end{proof}

\begin{corollary}
At high SNRs, the asymptotic expressions for the OP of T of the ambient backscatter NOMA systems can be expressed as

$\bullet$ Non-ideal conditions

For ideal conditions, the asymptotic expression for the OP of the BD in (20) is provided at the top of next page.
\begin{figure*}[!t]\label{18}
\normalsize
\begin{align}\nonumber
\label{15}
P_{out,\infty }^{T,ni} =&  \!-\! \frac{{{\lambda _{S{R_n}}}{\xi_{{R_n}}}{\gamma _{thc}^{{R_n}}}}}{{{\lambda _{ST}}{\lambda _{T{D_n}}}{\Delta _5}}}{e^{A_2^{{R_n}}}}\!\!\left( \! {\frac{\pi }{N}\sum\limits_{k = 0}^N {{e^{ \!-\! \left(\! {\frac{{2A_3^{{R_n}}}}{{A_4^{{R_n}}\left(\! {{\vartheta _k} + 1} \!\right)}} - \frac{{A_1^{{R_n}}A_4^{{R_n}}\left( {{\vartheta _k} + 1} \right)}}{2}} \!\right)}}\sqrt {1 \!-\! \vartheta _k^2} \left(\! {\frac{1}{{{\vartheta _k} \!+\! 3}} + \frac{1}{{{\vartheta _k} \!+\! 1}}} \!\right) \!-\! 2{K_0}\left( \!{2\sqrt { - A_1^{{R_n}}A_3^{{R_n}}} } \right)} } \right) \\
&+ \frac{{2{\lambda _{S{R_n}}}}}{{{\lambda _{T{R_n}}}{\lambda _{ST}}{\varsigma _{{R_n}}}{B_{{R_n}}}}}{e^{ - {B_5}}}\sum\limits_{v = 1}^\infty  {{{\left( { - 1} \right)}^v}\frac{1}{{B_4^v}}} {\left( {\frac{{{B_3}}}{{{B_1}}}} \right)^{\frac{v}{2}}}{K_v}\left( {2\sqrt {{B_3}{B_1}} } \right)  .
\end{align}
\hrulefill \vspace*{0pt}
\end{figure*}

$\bullet$ Ideal conditions
\begin{equation}
P_{out,\infty }^{T,id} = 1 + {\Delta _9}{e^{{\Delta _9}}}{\rm{Ei}}\left( { - {\Delta _9}} \right) - {\Delta _{11}}{e^{{\Delta _{11}}}}{\rm{Ei}}\left( { - {\Delta _{11}}} \right).
\end{equation}

\end{corollary}

Next, in order to obtain more insights, the diversity orders for $R_f$, $R_n$ and $T$ are investigated, which can be defined as \cite{30}:
\begin{equation}\label{18}
d =  - \mathop {\lim }\limits_{\gamma  \to \infty } \frac{{\log \left( {P_{out}^\infty } \right)}}{{\log {\gamma}}}.
\end{equation}

\begin{corollary}
The diversity orders of $D_f$, $D_n$ and $BD$ are given as:
\begin{equation}\label{21}
{d_{{R_f}}} = {d_{{R_n}}} =d_{T}^{id} = d_{T}^{ni}= 0.
\end{equation}

\end{corollary}

\begin{remark}
From \textbf{Corollary 1}-\textbf{Corollary 4}, we can obtain that: 1) RHIs, CEEs and ipSIC have detrimental effects on the reliability of the considered systems; 2) The asymptotic outage performance of the $R_f$, $R_n$ and $T$ become a constant when the transmit SNR is in a high state, indicating that there are error floors for the OP; 3) From Eq. (23), it can be observed that the diversity orders of the considered system are zero due to the fixed constant for the OP in the high SNR regime.
\end{remark}

\subsection{IP Analysis}

User $i$ ($i \in \{ {{R_f},{R_n},T} \}$) will be intercepted if $E$ can successfully wiretap $j$'s signal, i.e., ${\gamma _E^{p} > {\gamma _{thj}^E}}$, $p \in {\rm{ }}\left\{ {{x_2},{x_1},c\left( t \right)} \right\}$, $j \in \left( {2,1,c} \right)$. Thus, the IP of $i$ by $E$ can be expressed as
\begin{equation}
\label{11}
P_{int}^{i} = {{\rm{P}}_r}\left( {\gamma _E^{p} > \gamma _{thj}^E} \right),
\end{equation}
where ${\gamma _{thj}^E}$ is the secrecy SNR threshold of $i$.
\begin{theorem}
The analytical expressions for the IP of the far reader, the near reader and T can be respectively obtained as

For the far and near readers, we have
\begin{equation}
\label{11}
P_{int}^{{R_f}} \!\!=\!  -\! {\Delta _2^{{E}}}{e^{{\Delta _1^{{E}}} \!-\! {\Delta _3^{{E}}} \!-\! \frac{{{\gamma _{th2}^{{E}}}}}{{{\lambda _{S{E}}}\gamma \left( {{a_2} \!-\! {Q_{{E}}}{\gamma _{th2}^{{E}}}} \right)}}}}{\rm{Ei}}\left( { \!-\! {\Delta _1^{{E}}}} \right),
\end{equation}
\begin{equation}
\label{7}
P_{{\mathop{\rm int}} }^{{R_n}} = - {\Delta _{16}}{e^{{\Delta _{17}} - {\Delta _{18}} - \frac{{\gamma _{th2}^E}}{{{\lambda _{SE}}\gamma \left( {{a_1} - {O_E}\gamma _{th2}^E} \right)}}}}{\rm{Ei}}\left( { - {\Delta _{17}}} \right),
\end{equation}

\noindent where ${\Delta _{17}} \!\!=\!\!\! \left(\! {\frac{{{M_E}{\gamma _{th2}}}}{{{\lambda _{SE}}\!\left( \!{{a_1} \!- \!{O_E}\!\gamma _{th2}^E} \!\right)}} \!\!+\! \! \frac{1}{{{\lambda _{ST}}}}} \!\right)\!\!\frac{{{\lambda _{SE}}\!\left(\! {{a_1} \!- {O_E}\!\gamma _{th2}^E} \!\right)\! +\! {\lambda _{TE}}\!{C_E}\!\gamma _{th2}^E}}{{{\lambda _{TE}}\!{B_E}\!\gamma _{th2}^E}}$, ${\Delta _{16}} = \frac{{{\lambda _{SE}}\left( {{a_1} - {O_E}\gamma _{th2}^E} \right)}}{{{\lambda _{ST}}{\lambda _{TE}}{B_E}\gamma _{th2}^E}}$, and ${\Delta _{18}} = \frac{{{\psi _E}\gamma _{th2}^E}}{{{\lambda _{SE}}\left( {{a_1} - {O_E}\gamma _{th2}^E} \right)}} $.

For T, we have

$\bullet$ Non-ideal conditions

For non-ideal conditions, the analytical expression for the IP of T in (27) is at the top of next page.
\begin{figure*}[!t]\label{18}
\normalsize
\begin{align}\nonumber
\label{15}
P_{int}^{T,ni} =&  -\! \frac{{{\lambda _{SE}}{\xi_E}\gamma _{thc}^E}}{{{\lambda _{ST}}{\lambda _{TE}}\Delta _5^E}}{e^{A_2^E}}\!\left(\! {\frac{\pi }{N}\sum\limits_{k = 0}^N {{e^{ \!- \!\left( {\frac{{2\left( \!{A_3^E \!+\! \Delta _8^E} \!\right)}}{{A_4^E\left( {{\vartheta _k} + 1} \!\right)}} \!-\! \frac{{A_1^EA_4^E\left( {{\vartheta _k} + 1} \right)}}{2}} \!\right)}}\!\!\sqrt {1\! -\! \vartheta _k^2} \!\left(\! {\frac{1}{{{\vartheta _k} \!+\! 3}} \!- \!\frac{1}{{{\vartheta _k} \!+\! 1}}} \!\right)\! \!+\! 2{K_0}\!\left( {2\sqrt { - A_1^E\left( {A_3^E + \Delta _8^E} \right)} } \right)} } \right) \\
&+ 2\sqrt {{{\Delta _{15}}}{{{\Delta _{13}}}}} {e^{ - {\Delta _{14}}}}{K_1}\left( {2\sqrt {{\Delta _{13}}{\Delta _{15}}} } \right).
\end{align}
\hrulefill \vspace*{0pt}
\end{figure*}

\begin{table*}[!htb]
  \begin{center}
  \caption{Table of Parameters for numerical results.}
    \begin{tabular}{|l|l|}
      \hline
      Power sharing coefficients of NOMA & ${a_1} = 0.2$, ${a_2} = 0.8$ \\
      \hline
      Noise power & ${N_0} = 1$ \\
      \hline
      Reflection coefficient &  $\beta  = 0.1$ \\
      \hline
      ipSIC parameter & $\varepsilon  = 0.01$ \\
      \hline
      Power coefficient of artificial noise & $\varphi _J = 0.1$ \\
      \hline
      Interfering factor of readers & $\varpi  = 0.5$\\
      \hline
      RHIs parameter &  ${\kappa _{S{R_f}}} = {\kappa _{S{R_n}}} = {\kappa _{SE}} = \kappa=0.1 $ \\
      \hline
      Channel fading parameters & $\left\{ {{\lambda _{S{R_f}}},{\lambda _{S{R_n}}},{\lambda _{SB}},{\lambda _{SE}},{\lambda _{T{R_f}}},{\lambda _{T{R_n}}}{\lambda _{TE}}} \right\} = \left\{ {4,6,1,0.5,1,2,0.3} \right\}$ \\
      \hline
      CEEs parameter &  $\sigma _{{e_{S{R_f}}}}^2 = \sigma _{{e_{S{R_n}}}}^2 = \sigma _{{e_{SB}}}^2 = \sigma _{{e_{SE}}}^2 = \sigma _{{e_{T{R_f}}}}^2 = \sigma _{{e_{T{R_n}}}}^2 = \sigma _{{e_{TE}}}^2 = \sigma _e^2=0.05 $\\
      \hline
      Targeted data rates (OP) & $\left\{ {\gamma _{th1}^{{R_n}},\gamma _{th2}^{{R_n}} = \gamma _{th2}^{{R_f}},\gamma _{thc}^{{R_n}}} \right\} = \left\{ {1.2,1,0.001} \right\}$ \\
      \hline
      Targeted data rates (IP) & $\left\{ {\gamma _{th1}^E,\gamma _{th2}^E,\gamma _{thc}^E} \right\} = \left\{ {0.12,0.3,0.01} \right\}$ \\

      \hline
    \end{tabular}

  \end{center}
\end{table*}

\noindent In (27), ${\Delta _{13}} = 1/\left( {{\lambda _{ST}}{\lambda _{TE}}\Delta _5^E} \right)$, ${\Delta _{14}} = \frac{{{C_E}\gamma _{thc}^E\left( {{\lambda _{ST}} + {\lambda _{TE}}} \right)}}{{{\lambda _{TE}}{\lambda _{ST}}\Delta _5^E}}$, and ${\Delta _{15}} = \frac{{C_E^2{{\left( {\gamma _{thc}^E} \right)}^2}/\left( {{\lambda _{TE}}\Delta _5^E} \right)}}{{\Delta _5^E}} + \left( {{\psi _E} + \frac{1}{\gamma } } \right)\gamma _{thc}^E$.

$\bullet$ Ideal conditions

For ideal conditions, the analytical expression for the IP of T in (28) is at the top of next page.
\begin{figure*}[!t]\label{18}
\normalsize
\begin{align}\nonumber
\label{15}
P_{{\mathop{\rm int}} }^{T,id} =& 1 - \frac{{\pi \gamma _{thc}^E}}{{N{\lambda _{ST}}{\lambda _{TE}}\gamma \Delta _5^E}}\sum\limits_{k = 0}^N  \left( {{\vartheta _k} + 1} \right){K_0}\left( {\left( {{\vartheta _k} + 1} \right)\sqrt {\frac{{\gamma _{thc}^E}}{{{\lambda _{ST}}{\lambda _{TE}}\gamma \Delta _5^E}}} } \right)\sqrt {1 - \vartheta _k^2}  \\
&-\frac{2}{{{\lambda _{ST}}{\lambda _{TE}}}}{e^{\frac{1}{{{\lambda _{SE}}{\xi_E}\gamma }}}}\int_{\frac{{\gamma _{thc}^E}}{{\Delta _5^E}}}^\infty  {{e^{ - \frac{{\Delta _5^Ey}}{{{\lambda _{SE}}{\xi_E}\gamma _{thc}^E}}}}{K_0}\left( {2\sqrt {\frac{y}{{{\lambda _{ST}}{\lambda _{TE}}}}} } \right)dy} .
\end{align}
\hrulefill \vspace*{0pt}
\end{figure*}
\end{theorem}
\begin{proof}
See Appendix C.
\end{proof}

Moreover, for further investigation of the ambient backscatter NOMA secure communication systems, we also study the asymptotic behaviors of IP in the high MER region \cite{9032127}. MER is introduced to distinguish the channel state of the main link and eavesdropping link, being defined as ${\lambda _{me}} = \frac{{{\lambda _{ST}}}}{{{\lambda _{TE}}}}$.
\begin{corollary}
At high MERs, the asymptotic expression for the IP of $R_f$ of the ambient backscatter NOMA systems is given as
\begin{equation}
\label{10}
P_{int,\infty }^{{R_f}} \!=\!  -\! {\Delta _2}^{\!\prime} {e^{{\Delta _1}^{\!\prime} \! -\! {\Delta _3}^{\!\prime} \! -\! \frac{{\gamma _{th2}^E}}{{{\lambda _{\!SE}}\gamma \left(\! {{a_2} \!-\! {Q_E}\gamma _{\!th2}^E} \!\right)}}}}\!\!\left(\! {1 \!+\! {{b_1}^{\!\prime} }} \!\right)\!{\rm{Ei}}\!\left( \!{ -\! \left(\! {{\Delta _1}^{\!\prime} \! + \!{{b_1}^{\!\prime} }} \right)} \!\right),
\end{equation}
\noindent where ${\Delta _1}^{\!\prime}  \!\!=\!\!\! \frac{{{M_E}}}{{{\lambda _{\!T\!E}}\!{B_{\!E}}}} + \frac{{{M_E}{C_E}\gamma _{th2}^E}}{{{\lambda _{\!SE}}\left( \!{{a_2}\! -\! {Q_E}\!\gamma _{th2}^E} \!\right)\!{B_E}}}$, ${\Delta _2}^{\!\prime} \!=\! \frac{{{\lambda _{\!SE}}\left( \!{{a_2} \!-\! {Q_E}\gamma _{th2}^E} \right)}}{{{\lambda _{me}}\lambda _{TE}^2{B_E}\gamma _{th2}^E}}$, ${\Delta _3}^{\!\prime} \!=\! \frac{{{\psi _E}\gamma _{th2}^E}}{{{\lambda _{\!SE}}\!\left(\! {{a_2}\! -\! {Q_E}\gamma _{th2}^E} \right)}}$, and ${{b_1}^{\!\prime} } \!=\! \frac{{{\lambda _{\!SE}}\!\left( {{a_2} - {Q_E}\gamma _{th2}^E} \right) + {\lambda _{TE}}{C_E}\gamma _{th2}^E}}{{{\lambda _{me}}\lambda _{TE}^2{B_E}\gamma _{th2}^E}}$.
\end{corollary}

\begin{proof}
The proof follows by taking ${\lambda _{me}}$ large in (29) and simplifying the expressions by utilizing ${e^x} \approx 1 + x$ if $x \to 0$. Similarly, we can also obtain (30).
\end{proof}

\begin{corollary}
At high MERs, the asymptotic expression for the IP of $R_n$ of the ambient backscatter NOMA systems is given as
\begin{align}\nonumber
\label{10}
P_{int,\infty }^{{R_n}} =  - {\Delta _{16}}^\prime {e^{{\Delta _{17}}^\prime  - {\Delta _{18}}^\prime  - \frac{{\gamma _{th2}^E}}{{{\lambda _{SE}}\gamma \left( {{a_1} - {O_E}\gamma _{th2}^E} \right)}}}}\\
\times\left( {1 + {b_2}^\prime } \right){\rm{Ei}}\left( { - \left( {{\Delta _{17}}^\prime  + {b_2}^\prime } \right)} \right),
\end{align}
\noindent where ${\Delta _{17}}^\prime  \!\!=\!\!\! \frac{{{M_E}}}{{{\lambda _{\!T\!E}}\!{B_{\!E}}}} + \frac{{{M_E}{C_E}\gamma _{th2}^E}}{{{\lambda _{\!SE}}\left( \!{{a_1}\! -\! {O_E}\!\gamma _{th2}^E} \!\right)\!{B_E}}}$, ${\Delta _{16}}^\prime \!=\! \frac{{{\lambda _{\!SE}}\left( \!{{a_1} \!-\! {O_E}\gamma _{th2}^E} \right)}}{{{\lambda _{me}}\lambda _{TE}^2{B_E}\gamma _{th2}^E}}$, ${\Delta _{18}}^{\!\prime}\! =\! \frac{{{\psi _E}\gamma _{th2}^E}}{{{\lambda _{SE}}\left( \!{{a_1} \!- {Q_E}\gamma _{th2}^E} \!\right)}}$, and ${b_2}^{\!\prime} \!=\! \frac{{{\lambda _{SE}}\left(\! {{a_1} - {O_E}\gamma _{th2}^E} \!\right) \!+\! {\lambda _{\!TE}}{C_E}\gamma _{th2}^E}}{{{\lambda _{me}}\lambda _{TE}^2{B_E}\gamma _{th2}^E}}$.
\end{corollary}

\begin{corollary}
At high MERs, the asymptotic expression for the IP of $T$ of the ambient backscatter NOMA systems can be written by

$\bullet$ Non-ideal conditions

For non-ideal conditions, the asymptotic expression for the OP of the BD in (31) is provided at the top of next page.
\begin{figure*}[!t]\label{18}
\normalsize
\begin{align}\nonumber
\label{15}
&P_{int,\infty }^{T,ni} =  - \frac{{\pi {\lambda _{SE}}{\xi _E}\gamma _{thc}^E}}{{N{\lambda _{me}}\lambda _{TE}^2\Delta _5^E}}{e^{ - \frac{{{M_E}\gamma _{thc}^E}}{{{\lambda _{TE}}\Delta _5^i}}}}\left(\! {1 \!-\! \frac{{{C_E}\gamma _{thc}^E}}{{{\lambda _{me}}{\lambda _{TE}}\Delta _5^E}}} \!\right)\!\sum\limits_{k = 0}^N {{e^{ - \frac{{2\left( {A_3^E + \Delta _8^E} \right)}}{{A_4^E\left( {{\vartheta _k} + 1} \right)}}}}\!\left(\! {1 \!+\! \frac{{{A_1}^\prime A_4^E\left( {{\vartheta _k} + 1} \right)}}{2}} \!\right)\!\sqrt {1 \!-\! \vartheta _k^2} \left( {\frac{1}{{{\vartheta _k} + 3}} \!- \!\frac{1}{{{\vartheta _k} + 1}}} \right)}  \\
 & +\! \frac{{{\lambda _{\!SE}}\!{\xi _E}\gamma _{\!thc}^E}}{{{\lambda _{me}}\lambda _{TE}^2\Delta _5^E}}{e^{ \!- \frac{{{M_E}\gamma _{thc}^E}}{{{\lambda _{TE}}\Delta _5^i}}}}\!\!\left(\! {1\! - \!\frac{{{C_E}\gamma _{thc}^E}}{{{\lambda _{me}}{\lambda _{\!TE}}\Delta _5^{\!E}}}} \!\right)\!{\rm{ln}}\!\left(\!\! {\sqrt { \!-\! A_1^E\!\left(\! {A_3^E \!+\! \Delta _8^E} \!\right)} } \!\right)
\!\!+\! 2\sqrt {{\Delta _{15}}{\Delta _{13}}^{\!\prime} } {K_1}\!\left(\! {2\sqrt {{\Delta _{13}}^{\!\prime} \!{\Delta _{15}}} } \!\right){e^{\! - \frac{{{C_E}\gamma _{thc}^E}}{{{\lambda _{TE}}\Delta _5^E}}}}\left(\! {1 \!- \!\frac{{{C_E}\gamma _{thc}^E}}{{{\lambda _{me}}{\lambda _{\!TE}}\Delta _5^{\!E}}}} \!\right).
\end{align}
\hrulefill \vspace*{0pt}
\end{figure*}

\noindent In (31), ${A_1}^\prime  \;=  - \; \frac{1}{{{\lambda _{SE}}\;{\xi _E}\;{\lambda _{me}}\;\lambda _{TE}^2\Delta _5^E{\gamma ^2}}}$, ${\Delta _{13}}^\prime  = \frac{1}{{{\lambda _{me}}\;\lambda _{TE}^2\;\Delta _5^E}}$,\\
${K_1}\!\left( \!{2\sqrt {{\Delta _{13}}^\prime \!{\Delta _{15}}} } \!\right) \!\approx \! {I_1}\left(\! {2\sqrt {{\Delta _{13}}^\prime {\Delta _{15}}} } \!\right)\left( \! {{\rm{ln}}\left(\! {\sqrt {{\Delta _{13}}^\prime {\Delta _{15}}} } \right) \!+\! \upsilon } \!\right)+ \\
\frac{1}{2}{\left(\! {\sqrt {{\Delta _{13}}^{\!\prime} \!{\Delta _{15}}} } \!\right)^{ - 1}} \!\!- \frac{1}{2}\!\sum_{l = 0}^3 \!{\frac{{{{\left(\! {\sqrt {{\Delta _{13}}^{\!\prime} {\Delta _{15}}} } \!\right)}^{2l + 1}}}}{{l!\left( {l + 1} \right)}}} \!\!\left( \! {\sum_{k = 1}^l {\frac{1}{k} \!+\! \sum_{k = 1}^{l{\rm{ + 1}}} {\frac{1}{k}} } } \!\right)$.\footnote{For large MER,  in order to achieve a better approximation effect, we only need to consider the first three terms of $l$, i.e. $l=1,2,3$.}

$\bullet$ Ideal conditions

For Ideal conditions, the analytical expression for the IP of T in (32) is provided at the top of next page.
\begin{figure*}[!t]\label{18}
\normalsize
\begin{align}\nonumber
\label{15}
P_{int,\infty }^{T,id} =& 1 + \frac{{\pi \gamma _{thc}^E}}{{N{\lambda _{me}}\lambda _{TE}^2\gamma \Delta _5^E}}\sum\limits_{k = 0}^N  \left( {{\vartheta _k} + 1} \right){\rm{ln}}\left( {\frac{{{\vartheta _k} + 1}}{2}\sqrt {\frac{{\gamma _{thc}^E}}{{{\lambda _{me}}\lambda _{TE}^2\gamma \Delta _5^E}}} } \right)\sqrt {1 - \vartheta _k^2}  \\
&+\frac{2}{{{\lambda _{me}}\lambda _{TE}^2}}{e^{\frac{1}{{{\lambda _{SE}}{\xi_E}\gamma }}}}\int_{\frac{{\gamma _{thc}^E}}{{\Delta _5^E}}}^\infty  {{e^{ - \frac{{\Delta _5^Ey}}{{{\lambda _{SE}}{\xi_E}\gamma _{thc}^E}}}}{\rm{ln}}\left( {\sqrt {\frac{y}{{{\lambda _{me}}\lambda _{TE}^2}}} } \right)dy} .
\end{align}
\hrulefill \vspace*{0pt}
\end{figure*}

\end{corollary}
\begin{proof}
The proof follows by taking ${\lambda _{me}}$ large in (31) and (32) and simplifying the expressions by utilizing ${e^{-x}} \approx 1 - x$ and ${K_0}\left( x \right) \approx  - {\rm{ln}}\left( x \right)$ if $x \to 0$.
\end{proof}

\begin{figure}[!t]
\setlength{\abovecaptionskip}{0pt}
\centering
\includegraphics [width=3.5in]{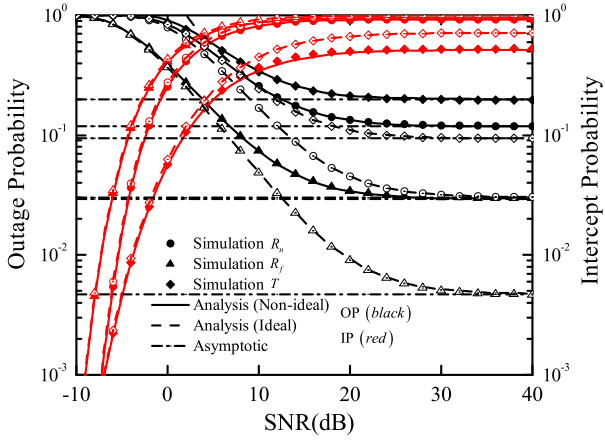}
\caption{{ OP and IP versus the transmit SNR. }}
\label{fig2}
\end{figure}

\begin{figure}[!t]
\setlength{\abovecaptionskip}{0pt}
\centering
\includegraphics [width=3.5in]{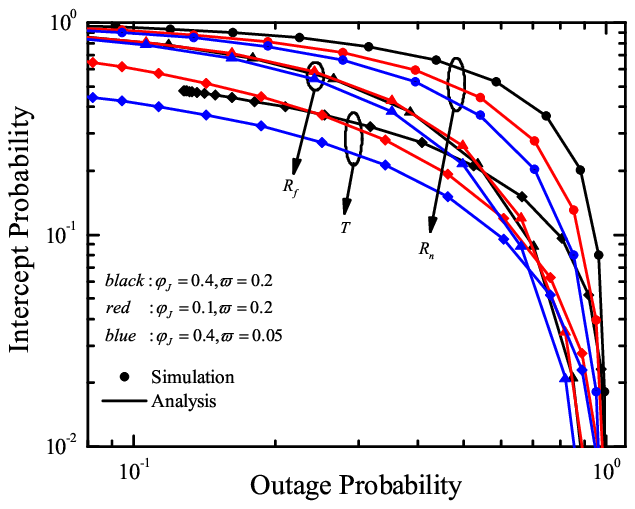}
\caption{{ IP versus OP for different power coefficient of artificial noise $\varphi _J$. }}
\label{fig3}
\end{figure}

\begin{remark}
From \textbf{Theorem 4} and \textbf{Corollary 5}-\textbf{Corollary 7}, the following observations can be inferred: 1) RHIs, CEEs and ipSIC can enhance the security of the ambient backscatter NOMA systems; 2) When the reflection coefficient $\beta $ increases, both $P_{int}^{{R_f}}$ and $P_{int}^{{R_n}}$ decrease, while $P_{int}^{T}$ increases; 3) Increasing $\varphi _J$ can reduce the IP, thereby improving the reliability-security trade-off of the considered systems; 4) as ${\lambda _{me}}$ grows, the security for $R_n$ and $R_f$ is improved, while the security for $T$ is reduced.
\end{remark}

\section{Numerical Results}
In this section, simulation results are provided to verify the correctness of our theoretical analysis in Section III. The results are verified by using Monte Carlo simulations with ${10^6}$ trials. Unless otherwise stated, we set the parameters as shown in Table I is at the top of the previous page.

\begin{figure}[t!]
\centering
\subfigure[OP versus RHIs and CEEs.]{
\begin{minipage}[t]{1\linewidth}
\centering
\includegraphics[width= 3.5in, height=2.652in]{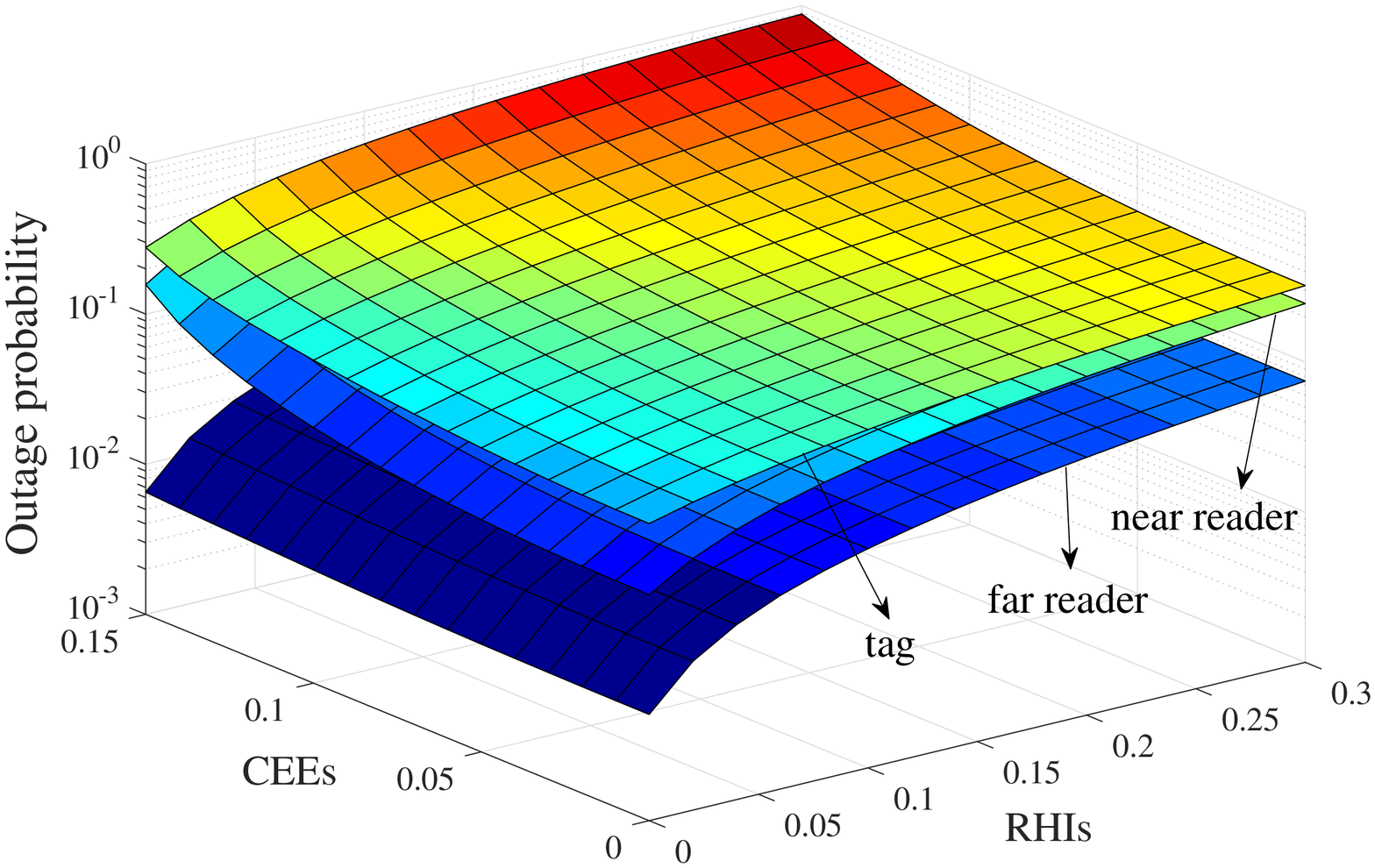}
\end{minipage}
}
\subfigure[IP versus RHIs and CEEs.]{
\begin{minipage}[t]{1\linewidth}
\centering
\includegraphics[width= 3.5in, height=2.652in]{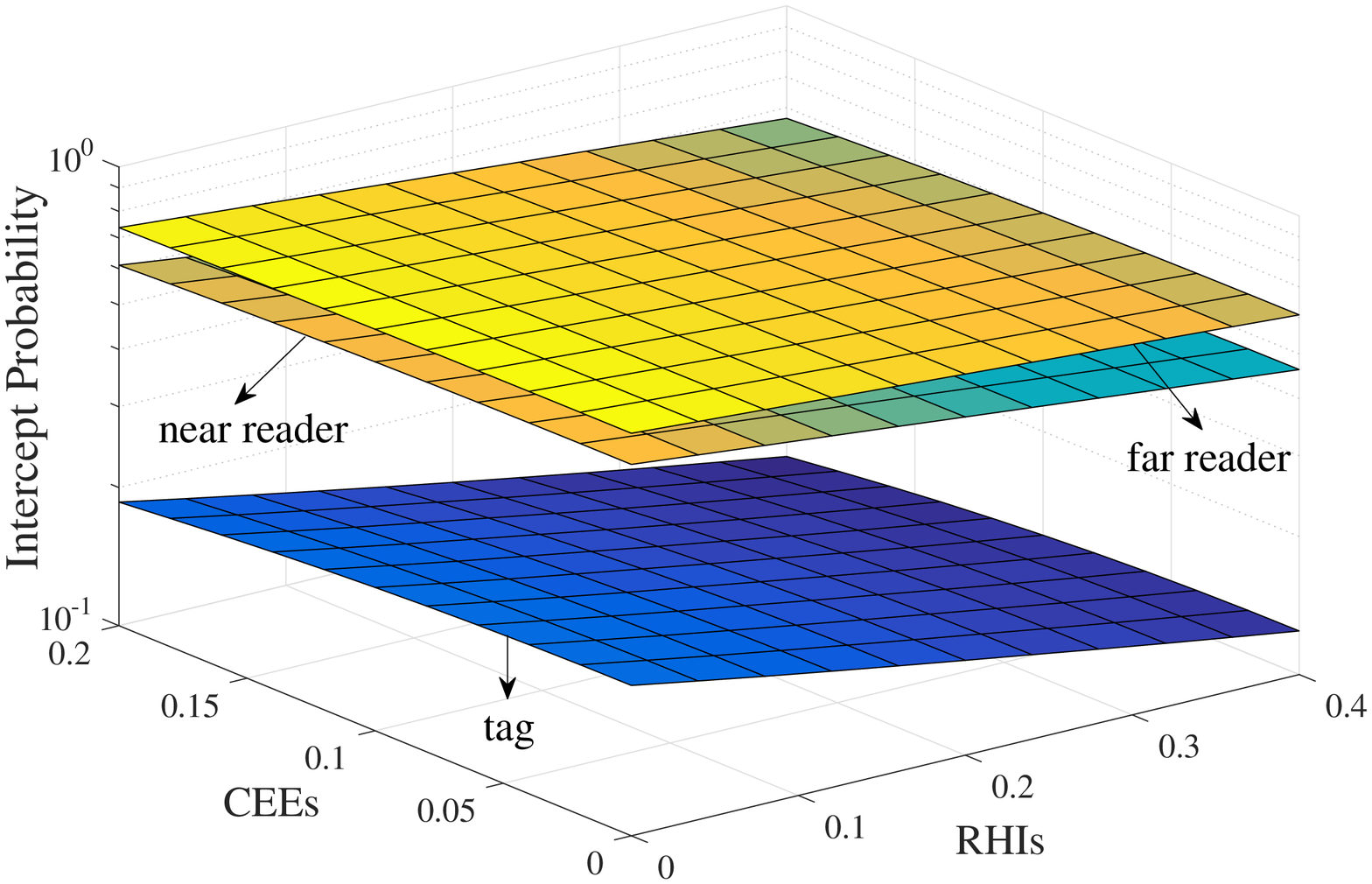}
\end{minipage}
}
\caption{{OP and IP versus RHIs and CEEs.}}
\label{fig4}
\end{figure}

\begin{figure}[t!]
\centering
\subfigure[OP versus the transmit SNR for different $\varepsilon $ and $\beta $.]{
\begin{minipage}[t]{1\linewidth}
\centering
\includegraphics[width= 3.5in]{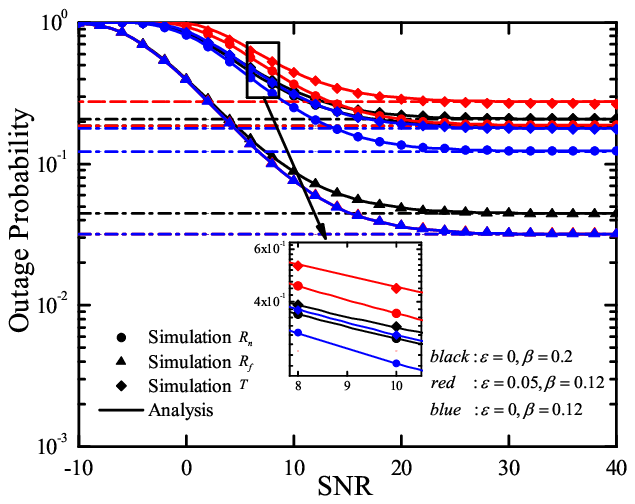}
\end{minipage}
}
\subfigure[IP versus the transmit SNR for different $\varepsilon $ and $\beta $.]{
\begin{minipage}[t]{1\linewidth}
\centering
\includegraphics[width= 3.5in]{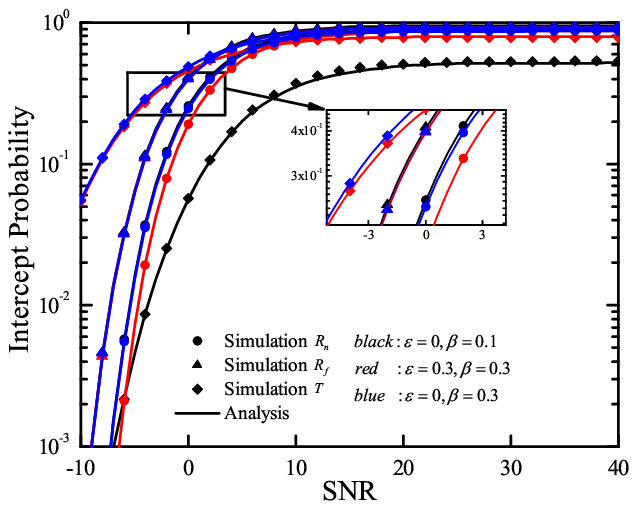}
\end{minipage}
}
\caption{{OP and IP versus the transmit SNR for different $\varepsilon $ and $\beta $.}}
\label{fig5}
\end{figure}
Fig. 2 plots the OP and the IP versus the transmit SNR for the far reader, the near reader and T, with $\kappa  = 0.1$ and $\sigma _e^2 = 0.05$. For comparison, the considered system performance of ideal conditions is provided with $\kappa  = 0$ and $\sigma _e^2 = 0$. It is shown that the theoretical results match well the simulations across the entire SNR region. We can also observe that the OP approaches a fixed constant due to the fixed estimation error and $\beta$ in the high SNR region, which results in zero diversity order. These results verify the conclusion in \textbf{Remark 1}. Moreover, RHIs have a positive impact on IP, which reveals that the ideal communication systems are more vulnerable to be eavesdropped than the non-ideal communication systems. Finally, we can also see that there exists a trade-off between reliability and security.

Fig. 3 demonstrates the impact of OP versus IP for different power coefficient of the artificial noise $\varphi _J$ and attenuation factor $\varpi $, with $\varphi _J = \left\{ {0.1,0.4} \right\}$ and $\varpi = \left\{ {0.2,0.05} \right\}$. In this simulation, we assume  $\kappa  = 0$ and $\sigma _e^2 = 0$. One can observe that as the power coefficient of the artificial noise $\varphi _J$ grows smaller, the reliability-security tradeoff of the considered system degrades significantly. This is because the interference signals at eavesdropper become more dominant, resulting in a higher IP. Similarly, the interference factor $\varpi $ of the readers increases so as to result in a higher OP, which indicates that the reliability-security tradeoff degrades obviously. It is noted that the IP of $T$ is the smallest, implying that $T$ has the most secure performance. Therefore, in order to improve reliability-security tradeoff of the system by artificial noise, the design with a greater power coefficient of the artificial noise and smaller interference factor of the reader is more important.

Fig. 4 presents the OPs andIPs versus RHIs $\kappa$ and CEEs $\sigma _e^2 $.  In this simulation, we set SNR $=25$ dB and $\varphi _J^{{R_n}}=0.05$ for the OP, while SNR $=5$ dB and $\varphi _J^{E}=0.2$ for IP. According to Fig.4 (a) and (b), it is clear that as $\kappa$ grows, $P_{out}^{{R_f}}$, $P_{out}^{{R_n}}$ and $P_{out}^{{T}}$ increase, while $P_{int}^{{R_f}}$, $P_{int}^{{R_n}}$ and $P_{int}^{{T}}$ decrease. Likewise, with an increasing $\sigma _e^2 $, the OPs of $R_f$, $R_n$ and $T$ increase, whereas those of the corresponding IPs decrease. It means that the reliability of $T$ is the worst, while it has better security. Moreover, for $R_f$, $R_n$ and $T$, the fluctuation for the OP and IP of RHIs is more obvious than that of CEEs, which shows that the reliability and security of the readers are more dependent on the ability of RHIs. Finally, we can also observe that as RHIs change, the OP of far reader changes drastically. In contrast, the change of OP of $T$ is the least obvious, most probably because $T$ eliminate part of interference caused by the far and near readers.

\begin{figure}[!t]
\setlength{\abovecaptionskip}{0pt}
\centering
\includegraphics [width=3.5in]{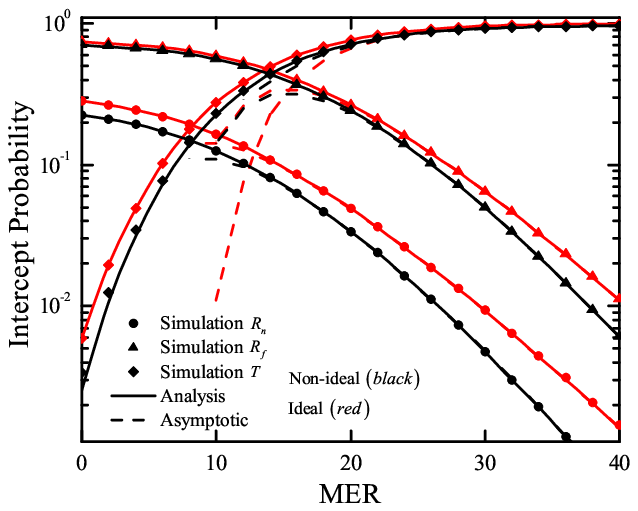}
\caption{{ IP versus MER for ideal and non-ideal conditions. }}
\label{fig6}
\end{figure}

Fig. 5 illustrates the OP and IP versus the transmit SNR for different $\varepsilon $ and $\beta $, respectively. In this simulation, we set: $\varepsilon  = \left\{ {0,0.05} \right\}$, $\beta  = \left\{ {0.2,0.12} \right\}$ for OP; $\varepsilon  = \left\{ {0,0.3} \right\}$, $\beta  = \left\{ {0.1,0.3} \right\}$ for IP. As can be seen in Fig. 5 (a), the error floors for the OP happen at high transmit SNR. The OP decreases as the transmit SNR increases, and depends on the value of $\varepsilon $ and $\beta $. More specifically, under perfect SIC $\left( {\varepsilon  = 0} \right)$, the outage behaviors of $R_f$, $R_n$ and $T$ improve remarkably when $\beta $ increases; similarly, for a fixed $\beta $, the increase of $\varepsilon $ also leads to lower reliability of $R_n$ and $T$.  By comparing Fig. 5 (a) with 5 (b), we can observe that $\varepsilon $ and $\beta $ have opposite effects on IP for the far reader, near reader, and $T$, while $\beta $ has identical effects on T, i.e., the increase of $\beta $ reduces the security of $T$. It is worth noting that OPs of $R_f$ and $R_n$ are more sensitive to $\beta $, which is due to the increase of interference from the backscatter link. For IP, $T$ is more sensitive to $\beta $. This happens because when $\beta $ increases, $E$ is more likely to eavesdrop the information of $c\left( t \right)$ successfully.

Fig. 6 presents the IP versus MER for $R_f$, $R_n$, and $T$ under ideal conditions with $\kappa  = 0$, $\sigma _e^2 = 0$, as well as non-ideal conditions with $\kappa  = 0.1$, $\sigma _e^2 = 0.05$.  In this simulation, we set SNR $=5$ dB, ${\lambda _{TE}}=2$, and $\left\{ {\gamma _{th1}^E,\gamma _{th2}^E,\gamma _{thc}^E} \right\} = \left\{ {0.3,0.3,1} \right\}$. From Fig. 6, we can observe that the asymptotic results are strict approximation of the IP in the high MER regime and the RHIs can enhance the security of $R_f$, $R_n$, and $T$. In addition, the IP of $R_f$ is much larger than that of $R_n$ when $R_f$ and $R_n$ have the same target rate, which is due to the fact that $R_f$ allocates more power. Therefore, considering the small power allocation coefficients $a_1$ and high target rate $\gamma _{th1}^E$ of the $R_n$, it is difficult for the information of $R_n$ to be eavesdropped by $E$. Finally, we can also observe that as MER grows, the security for $R_n$ and $R_f$ is improved, while the security for $T$ is reduced.

\section{Conclusion}
In this paper, we investigates the joint impacts of RHIs, CEEs and ipSIC on the reliability and the security of the ambient backscatter NOMA systems in terms of OP and IP. To improve the security performance, an artificial noise scheme was proposed, where the RF source simultaneously sends the signal and artificial noise to the readers and tag. The analytical expressions for the OP and the IP were derived. Furthermore, the asymptotic OP in the high SNR regime and the asymptotic IP in the high MER region are analyzed. The simulation results show that although RHIs, CEEs and ipSIC all have a significant negative impact for the OP of the far reader, near reader, and tag, they have a obvious positive effect for the IP on the three devices. In addition, the increase of $\beta$ will reduce the reliability and enhance the security for far reader and near reader. Finally, we can conclude that the optimal reliability-security tradeoff performance can be achieved by adjusting the power coefficient of the artificial noise and interference factor of the reader, which further drives ambient backscatter application in the IoT networks.

\numberwithin{equation}{section}
\section*{Appendix~A: Proof of Theorem 1} 
\renewcommand{\theequation}{A.\arabic{equation}}
\setcounter{equation}{0}
Substituting (3) into (6), the OP of $R_f$ can be expressed as
\begin{equation}
\label{6}
P_{out}^{{R_f}} = 1 - \underbrace {{{\rm{P}}_r}\left( {\gamma _{{R_f}}^{{x_2}} > \gamma _{th2}^{{R_f}}} \right)}_{{I_1}},
\end{equation}
where ${I_1}$ is calculated as follows:
\begin{align}\nonumber
\label{28_1}
&{I_1} = {{\rm{P}}_r}\left( {\gamma _{{R_f}}^{{x_2}} > \gamma _{th2}^{{R_f}}} \right)\\\nonumber
&= \int_{{\alpha _1}}^\infty  {\frac{1}{{{\lambda _{S{R_f}}}}}} {e^{ - \frac{x}{{{\lambda _{S{D_f}}}}}}}\frac{1}{{{\lambda _{T{R_f}}}}}{e^{ - \frac{y}{{{\lambda _{T{R_f}}}}}}}\frac{1}{{{\lambda _{ST}}}}{e^{ - \frac{z}{{{\lambda _{ST}}}}}}dxdydz\\\nonumber
&\mathop  = \limits^{u = z + \alpha } \Delta _2^{{R_f}}{e^{\Delta _1^{{R_f}} - \Delta _3^{{R_f}} - \frac{{\gamma _{th2}^{{R_f}}}}{{{\lambda _{S{R_f}}}\gamma \left( {{a_2} - {Q_{{R_f}}}\gamma _{th2}^{{R_f}}} \right)}}}}\int_\alpha ^\infty  {{e^{ - {\alpha _3}u}}\frac{1}{u}} du \\
& \mathop { = }\limits^{{l_1}} 1 \!+\! \Delta _2^{{R_f}}{e^{\Delta _1^{{R_f}}\! -\! \Delta _3^{{R_f}} \!-\! \frac{{\gamma _{th2}^{{R_f}}}}{{{\lambda _{S{R_f}}}\gamma \left( {{a_2} \!-\! {Q_{{R_f}}}\gamma _{th2}^{{R_f}}} \right)}}}}{\rm{Ei}}\left( { \!-\! \Delta _1^{{R_f}}} \right),
\end{align}
where ${\alpha _1} = \frac{{\left( {{B_{{R_f}}}z + {C_{{R_f}}}} \right)\gamma \gamma _{th2}^{{R_f}}y + {M_{{R_f}}}\gamma \gamma _{th2}^{{R_f}}z + \left( {{\psi _{{R_f}}} + 1} \right)\gamma \gamma _{th2}^{{R_f}}}}{{\left( {{a_2} - {Q_{{R_f}}}{\gamma _{th2}}} \right)\gamma }}$, ${\alpha _2} = \frac{{{\lambda _{S{R_f}}}\left( {{a_2} - {Q_{{R_f}}}\gamma _{th2}^{{R_f}}} \right) + {\lambda _{T{R_f}}}{C_{{R_f}}}\gamma _{th2}^{{R_f}}}}{{{\lambda _{T{R_f}}}{B_{{R_f}}}\gamma _{th2}^{{R_f}}}}$, ${\alpha _3} = \frac{{{M_{{R_f}}}\gamma _{th2}^{{R_f}}}}{{{\lambda _{S{R_f}}}\left( {{a_2} - {Q_{{R_f}}}{\gamma _{th2}}} \right)}} + \frac{1}{{{\lambda _{ST}}}}$, and the step ${l_1}$ is obtained by utilizing \cite[Eq. (3.352)]{32}. Finally, substituting (A.2) into (A.1), we can obtain (7).

Similarly, substituting (3) and (4) into (11), the (12) can be obtained.

\numberwithin{equation}{section}
\section*{Appendix~B: Proof of Theorem 3} 
\renewcommand{\theequation}{B.\arabic{equation}}
\setcounter{equation}{0}
Substituting (3), (4) and (5) into (14), the OP of $T$ can be expressed as
\begin{equation}
\label{6}
P_{out}^T \!=\! 1 \!-\! \underbrace {{{\rm{P}}_r}\left( {\gamma _{{R_n}}^{{x_2}} \!>\! \gamma _{th2}^{{R_n}},\gamma _{{R_n}}^{{x_1}} > \gamma _{th1}^{{R_n}},\gamma _{{R_n}}^{c\left( t \right)} > \gamma _{thc}^{{R_n}}} \right)}_{{I_2}},
\end{equation}

$\bullet$ Non-ideal conditions

For non-ideal conditions, ${I_2}$ is calculated as (A.2), shown at the top of the next page.
\begin{figure*}[!t]\label{18}
\normalsize
\begin{align}\nonumber
\label{15}
{I_2} \!=& {{\rm{P}}_r}\!\!\left(\!\!\! {{\varsigma _{{R_n}}}\!\gamma\!\! \left[\! {\left( \!{{B_{{\!R\!_n}}}\!{{\left|\! {{{\hat h}_{\!S\!T}}} \!\right|}^2} \!\!\!+ \!\! {C_{\!{R\!_n}}}}\!\! \right)\!{{\left|\! {{{\hat h}_{\!T{\!R\!_n}}}} \!\right|}^2} \!\!\!+\!\! {M_{{\!R\!_n}}}\!{{\left|\! {{{\hat h}_{\!ST}}} \!\right|}^2} \!\!\!+\!\! {\psi _{{\!R_n}}} \!\!+\!\! \frac{1}{\gamma }} \!\right]\! \!< \!{{\left| \!{{{\hat h}_{\!S{\!R_n}}}} \!\right|}^2}\! \!<\! \frac{{\!\left(\!\! {\Delta _5^{{\!R_n}}\!{{\left| \!{{{\hat h}_{\!ST}}} \!\right|}^2}\!\!\! -\! {C_{{\!R_n}}}\!\!\gamma _{\!thc}^{{\!R_n}}} \!\right)\!{{\left| \!{{{\hat h}_{\!T{\!R_n}}}} \!\right|}^2}\!\!\! -\!\! {M_{{\!R_n}}}\!{{\left|\! {{{\hat h}_{\!ST}}}\! \right|}^2}\!\!\gamma _{\!thc}^{{\!R_n}}\!\! - \!\!\left( \!\!{{N_{{\!R_n}}} \!\!\!+ \!\! \frac{1}{\gamma }} \!\!\right)\!\!\gamma _{\!thc}^{{\!R_n}}}}{{{\xi _{{R_n}}}\gamma _{thc}^{{R_n}}}}} \!\!\right)\\\nonumber
 =& \int_{\frac{{{C_{{R_n}}}\gamma _{thc}^{{R_n}}}}{{\Delta _5^{{R_n}}}}}^\infty  \!\!{\int_{\frac{{{M_{{R_n}}}\gamma _{thc}^{{R_n}}z \!+\! \left( {{\psi _{{\!R_n}}}\! + \! \frac{1}{\!\gamma }} \!\right)\!\gamma _{thc}^{{R_n}}}}{{\Delta _5^{{R_n}}z \!-\! {C_{{R_n}}}\gamma _{thc}^{{R_n}}}}}^\infty \!\! {\int_{{\varsigma _{{R_n}}}\left[ {\left( {{B_{{R_n}}}z \!+\! {C_{{R_n}}}} \right)y \!+\! {M_{{R_n}}}z \!+\! {\psi _{{R_n}}}\! +\! \frac{1}{\gamma }} \right]}^{\frac{{\left( {\Delta _5^{{R_n}}z \!-\! {C_{{R_n}}}\gamma _{thc}^{{R_n}}} \right)y \!- \!{M_{{R_n}}}{\gamma _{thc}}z \!-\! \left( {{\psi _{{R_n}}} \!+ \!1/\!\gamma } \!\right)\gamma _{thc}^{{R_n}}}}{{{\xi _{{R_n}}}\gamma _{thc}^{{R_n}}}}} \!\!{\frac{1}{{{\lambda _{S{R_n}}}}}{e^{ \!- \frac{x}{{{\lambda _{S{R_n}}}}}}}\!\frac{1}{{{\lambda _{\!T{R_n}}}}}{e^{ - \frac{y}{{{\lambda _{\!T{R_n}}}}}}}\frac{1}{{{\lambda _{\!ST}}}}{e^{ - \frac{z}{{{\lambda _{ST}}}}}}dxdydz} } } \\\nonumber
= &\underbrace {\int_{\frac{{{C_{{R_n}}}\gamma _{thc}^{{R_n}}}}{{\Delta _5^{{R_n}}}}}^\infty  {\int_{\frac{{{M_{{R_n}}}\gamma _{thc}^{{R_n}}z + \left( {{\psi _{{R_n}}} + \frac{1}{\gamma }} \right)\gamma _{thc}^{{R_n}}}}{{\Delta _5^{{R_n}}z - {C_{{R_n}}}\gamma _{thc}^{{R_n}}}}}^\infty  {{e^{ - \frac{1}{{{\lambda _{S{R_n}}}}}{\varsigma _{{R_n}}}\left[ {\left( {{B_{{R_n}}}z + {C_{{R_n}}}} \right)y + {M_{{R_n}}}z + {\psi _{{R_n}}} + \frac{1}{\gamma }} \right]}}\frac{1}{{{\lambda _{T{R_n}}}}}{e^{ - \frac{y}{{{\lambda _{T{R_n}}}}}}}\frac{1}{{{\lambda _{ST}}}}{e^{ - \frac{z}{{{\lambda _{ST}}}}}}dydz} } }_{{I_{21}}}\\
 & - \underbrace {\int_{\frac{{{C_{{R_n}}}\gamma _{thc}^{{R_n}}}}{{\Delta _5^{{R_n}}}}}^\infty  {\int_{\frac{{{M_{{R_n}}}\gamma _{thc}^{{R_n}}z + \left( {{\psi _{{R_n}}} + \frac{1}{\gamma }} \right)\gamma _{thc}^{{R_n}}}}{{\Delta _5^{{R_n}}z - {C_{{R_n}}}\gamma _{thc}^{{R_n}}}}}^\infty  {{e^{ - \frac{{\left( {\Delta _5^{{R_n}}z - {C_{{R_n}}}\gamma _{thc}^{{R_n}}} \right)y - {M_{{R_n}}}\gamma _{thc}^{{R_n}}z - \left( {{\psi _{{R_n}}} + 1/\gamma } \right)\gamma _{thc}^{{R_n}}}}{{{\lambda _{S{R_n}}}{\xi _{{D_n}}}\gamma _{thc}^{{R_n}}}}}}\frac{1}{{{\lambda _{T{R_n}}}}}{e^{ - \frac{y}{{{\lambda _{T{R_n}}}}}}}\frac{1}{{{\lambda _{ST}}}}{e^{ - \frac{z}{{{\lambda _{ST}}}}}}dydz} } }_{{I_{22}}}.
\end{align}
\hrulefill \vspace*{0pt}
\end{figure*}

By using some mathematical manipulations, we can obtain
\begin{align}\nonumber
\label{28_1}
&{I_{21}} = \int_{\frac{{{C_{{R_n}}}\gamma _{thc}^{{R_n}}}}{{\Delta _5^{{R_n}}}}}^\infty  {{\alpha _5}{e^{ - \frac{{{\varsigma _{{R_n}}}\left( {{M_{{R_n}}}z + {\psi _{{R_n}}} + \frac{1}{\gamma }} \right)}}{{{\lambda _{S{R_n}}}}} - {\alpha _4}}}\frac{1}{{{\lambda _{ST}}}}{e^{ - \frac{z}{{{\lambda _{ST}}}}}}dz}  \\\nonumber
&  = \int_0^\infty  {\frac{{{\lambda _{S{R_n}}}}}{{{\lambda _{T{R_n}}}{\lambda _{ST}}{\varsigma _{{R_n}}}{B_{{R_n}}}}}{e^{ - {\alpha _6}}}\frac{1}{{u + {B_4}}}{e^{ - \left( {{B_1}u + \frac{{{B_3} + {\Delta _6}}}{u}} \right)}}du}  \\\nonumber
& = \!\frac{{{\lambda _{S{R_n}}}}}{{{\lambda _{T{R_n}}}\!{\lambda _{ST}}\varsigma {T_{{R_n}}}}}{e^{ \!-\! {B_6}}}\!\sum\limits_{v = 1}^\infty  \!{{{\left(\! {\! -\! 1} \right)}^v}\!\frac{1}{{B_4^n}}} \!\!\int_0^\infty \!\! {{u^{v \!-\! 1}}{e^{\! -\! \left(\! {{B_1}u + \frac{{{B_3} + {\Delta _6}}}{u}} \!\right)}}du}  \\
&  \mathop  = \limits^{{l_2}} \!\!\!\frac{{2\!{\lambda _{S{R_n}}}}}{{{\lambda _{T{R_n}}}\!{\lambda _{ST}}\!{\varsigma _{{R_n}}}\!{B_{{R_n}}}}}\!{e^{ \!-\! {\alpha _6}}}\!\!\!\sum\limits_{v \!= \!1}^\infty  \!\!{{{\left( \!{ \!-\! 1} \right)}^v}\!\!\frac{1}{{B_4^v}}} \!\!{\left( \!\!{\frac{{\left( \!\!{{B_3} \!\!+\!\! {\Delta _6}}\! \right)}}{{{B_1}}}} \!\!\!\right)^{\frac{v}{2}}}\!\!\!{K_v}\!\!\left(\!\! {2\!\sqrt {\!\left(\! {{B_3}\!\! +\!\! {\Delta _6}} \!\!\right)\!{B_1}} } \!\!\right),
\end{align}
where $u = {\lambda _{S{R_n}}}{\lambda _{T{R_n}}}\Delta _5^{{R_n}}z - {\lambda _{S{R_n}}}{\lambda _{T{R_n}}}{C_{{R_n}}}\gamma _{thc}^{{R_n}}$, ${\alpha _4} = \frac{{\left[ {{\lambda _{T{R_n}}}{\varsigma _{{R_n}}}\left( {{B_{{R_n}}}z + {C_{{R_n}}}} \right) + {\lambda _{S{R_n}}}} \right]\left( {{\psi _{{R_n}}} + \frac{1}{\gamma }} \right)\gamma _{thc}^{{R_n}} + {M_{{R_n}}}\gamma _{thc}^{{R_n}}z}}{{{\lambda _{S{R_n}}}{\lambda _{T{R_n}}}\left( {\Delta _5^{{R_n}}z - {C_{{R_n}}}\gamma _{thc}^{{R_n}}} \right)}}$, ${\alpha _5} = \frac{{{\lambda _{S{R_n}}}}}{{{\lambda _{T{R_n}}}{\varsigma _{{R_n}}}\left( {{B_{{R_n}}}z + {C_{{R_n}}}} \right) + {\lambda _{S{R_n}}}}}$, ${\alpha _6} = {B_5} + \frac{{{\lambda _{T{R_n}}}{\varsigma _{{R_n}}}{B_{{R_n}}}\gamma _{thc}^{{R_n}}}}{{{\lambda _{S{R_n}}}{\lambda _{T{R_n}}}\gamma {\Delta _5}}} + \frac{{{\varsigma _{{R_n}}}}}{{{\lambda _{S{R_n}}}\gamma }}$, and ${l_2}$ is obtained by utilizing \cite[Eq. (3.471)]{12}.
\begin{align}\nonumber
\label{28_1}
&{I_{22}} \!\!\!=\!\!\!\! \int_{\frac{{{C_{\!{R_n}}}\!\!\gamma _{\!thc}^{\!{R_n}}}}{{\Delta _5^{{R_n}}}}}^\infty \!\!\!  {e^{\frac{{\!{\psi _{\!{R_n}}}\! +\! \frac{1}{\gamma }}}{{{\lambda _{\!S\!{R\!_n}}}\!\!{\xi _{\!{R\!_n}}}}} \!+\! \left( \!\!{\frac{{{M_{\!{R_n}}}}}{{{\lambda _{\!S{\!R\!_n}}}\!\!{\xi _{\!{R\!_n}}}}}\! -\! \frac{1}{{{\lambda _{\!S\!T}}}}} \!\!\right)\!z\!{\rm{ -\! }}{\alpha\! _7}}}\!\!\frac{{{\lambda _{S{R_n}}}{\xi _{{R_n}}}\gamma _{thc}^{{R_n}}}}{{{\lambda _{\!S\!T}}\!\!\left(\!\! {{\lambda _{\!T{\!R\!_n}}}\!\Delta _5^{{R_n}}\!z \!\!+\!\! \Delta _7^{{R_n}}}\!\! \right)}}\!dz     \\\nonumber
& \!\!=\! \frac{{{\lambda _{S{R_n}}}\!{\xi _{{R_n}}}\!\gamma _{thc}^{{R_n}}}}{{{\lambda _{ST}}{\lambda _{T{R_n}}}\Delta _5^{{R_n}}}}\!{e^{A_2^{{R_n}}}}\!\!\int_o^\infty \!\! {{e^{ \!-\! \left(\! {\left( \!{ \!-\! A_1^{{R_n}}} \!\right)u \!+\! \frac{{A_3^{{R_n}} \!+\! \Delta _8^{{R_n}}}}{u}} \!\!\right)}}\frac{1}{{u \!+\! A_4^{{R_n}}}}du}   \\ \nonumber
&= \frac{{{\lambda _{S{R_n}}}{\xi _{{R_n}}}\gamma _{thc}^{{R_n}}}}{{{\lambda _{ST}}{\lambda _{T{R_n}}}\Delta _5^{{R_n}}}}{e^{A_2^{{R_n}}}}\!\!\left[\! {\underbrace {\int_o^{A_4^{{R_n}}} \!{{e^{ \!- \left( \!{\left( { - {A_5}} \right)u + \frac{{{A_3}}}{u}} \right)}}\frac{1}{{u + A_4^{{R_n}}}}du} }_{{l_3}}} \right. \\
&\;\;\;\;\;\;+ \left. {\underbrace {\int_{A_4^{{R_n}}}^\infty  {{e^{ - \left( {\left( { - A_1^{{R_n}}} \right)u + \frac{{A_3^{{R_n}} + \Delta _8^{{R_n}}}}{u}} \right)}}\frac{1}{{u + A_4^{{R_n}}}}du} }_{{l_4}}} \right],
\end{align}
where $u = {\lambda _{S{R_n}}}{\xi _{{R_n}}}{\lambda _{T{R_n}}}\Delta _5^{{R_n}}z - {\lambda _{S{R_n}}}{\xi _{{R_n}}}{\lambda _{T{R_n}}}{C_{{R_n}}}\gamma _{thc}^{{R_n}}$, ${\alpha _7} = \frac{{\left( {{\lambda _{T{R_n}}}\Delta _5^{{R_n}}z + \Delta _7^{{R_n}}} \right)\left( {{\psi _{{R_n}}} + 1/\gamma  + {M_{{R_n}}}z} \right)}}{{{\lambda _{S{R_n}}}{\xi _{{R_n}}}{\lambda _{T{R_n}}}\left[ {\Delta _5^{{R_n}}z - {C_{{R_n}}}\gamma _{thc}^{{R_n}}} \right]}}$, and $l_3$ can be approximated by the Gaussian-Chebyshev quadrature \cite{33}, i.e., ${l_3} \approx  \frac{\pi }{N}\sum_{k = 0}^N {\frac{1}{{\left( {{\vartheta _k} + 3} \right)}}{e^{ - \left( {\frac{{2\left( {A_3^{{R_n}} + \Delta _8^{{R_n}}} \right)}}{{A_4^{{R_n}}\left( {{\vartheta _k} + 1} \right)}} - \frac{{A_1^{{R_n}}A_4^{{R_n}}\left( {{\vartheta _k} + 1} \right)}}{2}} \right)}}\sqrt {1 - \vartheta _k^2} }  $. Next, due to ${A_4^{{R_n}} \le 1}$, $l_4$ can be expressed as
\begin{align}\nonumber
\label{28_1}
{l_4} \approx & \int_{A_4^{{R_n}}}^\infty  {{e^{ - \left( {\left( { - A_1^{{R_n}}} \right)u + \frac{{A_3^{{R_n}} + \Delta _8^{{R_n}}}}{u}} \right)}}\frac{1}{u}du}      \\\nonumber
= &\int_0^\infty  {{e^{ - \left( {\left( { - A_1^{{R_n}}} \right)u + \frac{{A_3^{{R_n}} + \Delta _8^{{R_n}}}}{u}} \right)}}\frac{1}{u}du} \\\nonumber
 &- \int_0^{A_4^{{R_n}}} {{e^{ - \left( {\left( { - A_1^{{R_n}}} \right)u + \frac{{A_3^{{R_n}} + \Delta _8^{{R_n}}}}{u}} \right)}}\frac{1}{u}du}   \\ \nonumber
 =& 2{K_0}\left( {2\sqrt { - A_1^{{R_n}}\left( {A_3^{{R_n}} + \Delta _8^{{R_n}}} \right)} } \right) -\\
  \frac{\pi }{N}&\!\!\sum\limits_{k = 0}^N  \!\frac{1}{{{\vartheta _k} \!+\! 1}}{e^{ \!- \left(\! {\frac{{2\left(\! {A_3^{{R_n}} \!+\! \Delta _8^{{R_n}}} \!\right)}}{{A_4^{{R_n}}\left(\! {{\vartheta _k} \!+\! 1} \!\right)}} \!-\! \frac{{A_1^{{R_n}}A_4^{{R_n}}\left(\! {{\vartheta _k} \!+\! 1} \!\right)}}{2}} \!\right)}}\sqrt {1 \!-\! \vartheta _k^2} .
\end{align}
By substituting $l_3$ and (B.5) into (B.4), $I_{22}$ can be obtained; substituting (B.3) and (B.4) into (B.2), $I_2$ can be derived.

$\bullet$ Ideal conditions

Substituting $\kappa  = 0$ and $\sigma _e^2 = 0$ into (3), (4) and (5), ${C_{{R_f}}}={M_{{R_f}}}={C_{{R_n}}}={M_{{R_n}}}=0$. Then, the OP of $T$ under ideal conditions is given at the top of next page.
\begin{figure*}[!t]\label{18}
\normalsize
\begin{align}\nonumber
\label{28_1}
&P_{out}^{T,id} =1-{\int_{\frac{{\left( {{\psi _{{R_n}}} + 1/\gamma } \right)\gamma _{thc}^{{R_n}}}}{{\Delta _5^{{R_n}}}}}^\infty  {\left( {{e^{ - \frac{{{\varsigma _{{R_n}}}\left( {{B_{{R_n}}}y + {\psi _{{R_n}}}} \right)}}{{{\lambda _{S{R_n}}}}}}} - {e^{ - \frac{{\Delta _5^{{R_n}}y - \left( {{\psi _{{R_n}}} + 1/\gamma } \right)\gamma _{thc}^{{R_n}}}}{{{\lambda _{S{R_n}}}{\xi _{{R_n}}}\gamma _{thc}^{{R_n}}}}}}} \right)} \frac{2}{{{\lambda _{ST}}{\lambda _{T{R_n}}}}}{K_0}\left( {2\sqrt {\frac{y}{{{\lambda _{ST}}{\lambda _{T{R_n}}}}}} } \right)dy}   \\ \nonumber
&= 1-\underbrace {\int_0^\infty  {\left( {{e^{ - \frac{{{\varsigma _{{R_n}}}\left( {{B_{{R_n}}}y + {\psi _{{R_n}}}} \right)}}{{{\lambda _{S{R_n}}}}}}} - {e^{ - \frac{{\Delta _5^{{R_n}}y - \left( {{\psi _{{R_n}}} + 1/\gamma } \right)\gamma _{thc}^{{R_n}}}}{{{\lambda _{S{R_n}}}{\xi _{{R_n}}}\gamma _{thc}^{{R_n}}}}}}} \right)} \frac{2}{{{\lambda _{ST}}{\lambda _{T{R_n}}}}}{K_0}\left( {2\sqrt {\frac{y}{{{\lambda _{ST}}{\lambda _{T{R_n}}}}}} } \right)dy}_{{I_{31}}}\\
& + \underbrace {\int_0^{\frac{{\left( {{\psi _{{R_n}}} + 1/\gamma } \right)\gamma _{thc}^{{R_n}}}}{{\Delta _5^{{R_n}}}}} {\left( {{e^{ - \frac{{{\varsigma _{{R_n}}}\left( {{B_{{R_n}}}y + {\psi _{{R_n}}}} \right)}}{{{\lambda _{S{R_n}}}}}}} - {e^{ - \frac{{\Delta _5^{{R_n}}y - \left( {{\psi _{{R_n}}} + 1/\gamma } \right)\gamma _{thc}^{{R_n}}}}{{{\lambda _{S{R_n}}}{\xi _{{R_n}}}\gamma _{thc}^{{R_n}}}}}}} \right)} \frac{2}{{{\lambda _{ST}}{\lambda _{T{R_n}}}}}{K_0}\left( {2\sqrt {\frac{y}{{{\lambda _{ST}}{\lambda _{T{R_n}}}}}} } \right)dy}_{{I_{32}}}.
\end{align}
\hrulefill \vspace*{0pt}
\end{figure*}

\noindent{In (B.6), ${{I_{31}}}$ can be obtained by utilizing \cite[Eq. (6.611)]{32}, ${{I_{32}}}$ can be approximated by the Gaussian-Chebyshev quadrature \cite{33}. Thus, ${{I_{31}}}$ and ${{I_{32}}}$ can be expressed as}
\begin{align}
\label{28_1}
{I_{31}} \!\!\!=\!\! {\Delta _{11}}\!{e^{{\Delta _{11}}\! +\! \frac{1}{{{\lambda _{S{R_n}}}\!\gamma \!{\xi _{{R_n}}}}}}}\!{\rm{Ei}}\!\left( \!{ - \!{\Delta _{11}}} \!\right)\! \!\!-\!\! {\Delta _9}\!{e^{{\Delta _9} \!-\! \frac{{{\varsigma _{{R_n}}}}}{{{\lambda _{S{R_n}}}\!\gamma }}}}\!{\rm{Ei}}\!\left( \!\!{ -\! {\Delta _9}} \!\right),
\end{align}
\begin{align}\nonumber
\label{28_1}
{I_{32}} &= \frac{{\gamma _{thc}^{{R_n}}\pi }}{{N{\lambda _{T{R_n}}}{\lambda _{ST}}\gamma \Delta _5^{{R_n}}}}\sum\limits_{k = 0}^N {{K_0}\left( {2\sqrt {{\Delta _{10}}} } \right)} \sqrt {1 - \vartheta _k^2}\times \\
 & \left[\! {{e^{ \!-\! \left( {{\varsigma _{{R_n}}}{B_{{R_n}}}{\Delta _{10}}{\rm{ \!+\! }}\frac{{{\varsigma _{{R_n}}}}}{{{\lambda _{S{R_n}}}\gamma }}} \!\right)}}\! -\! {e^{\frac{1}{{{\lambda _{S{R_n}}}\gamma {\xi _{{R_n}}}}}\! -\! \frac{{{\vartheta _k} + 1}}{{2{\lambda _{S{R_n}}}\gamma {\xi _{{R_n}}}}}}}} \!\!\right] .
\end{align}

Similarly, substituting (B.7) and (B.8) into (B.6), we can obtain $P_{out}^{T,id}$.

\numberwithin{equation}{section}
\section*{Appendix~C: Proof of Theorem 4} 
\renewcommand{\theequation}{C.\arabic{equation}}
\setcounter{equation}{0}
According to $I_1$, we can obtain $P_{int}^{R_f}$ and $P_{int}^{R_n}$. Substituting (5) into (24), the IP of $T$ can be expressed as

$\bullet$ Non-ideal conditions

\begin{align}\nonumber
\label{15}
&P_{int}^{T,ni} \!\!\!=\!\!\! \int_{\frac{{{C_E}\!\gamma _{thc}^E}}{{\Delta _5^E}}}^\infty \!\!\! {\int_{\frac{{{M_E}\!\gamma _{thc}^Ez \!+ \!\left(\! {{\psi _E} \!+\! \frac{1}{\gamma }} \!\right)\!\gamma _{thc}^E}}{{\Delta _5^Ez\! -\! {C_E}\gamma _{thc}^E}}}^\infty \!\! {\frac{1}{{{\lambda _{TE}}\!{\lambda _{ST}}}}\!{e^{\! - \! \left(\! {\frac{y}{{{\lambda _{TE}}}} \!+\! \frac{z}{{{\lambda _{ST}}}}} \!\right)\!}}\!dydz} }  -  \\\nonumber
&\int_{\frac{{{C_E}\!\gamma _{thc}^E}}{{\Delta _5^E}}}^\infty \!\!\! {\int_{\frac{{{M_E}\!\gamma _{thc}^Ez \!+\! \left(\! {{\psi _E} \!+\! \frac{1}{\gamma }} \!\right)\!\gamma _{thc}^E}}{{\Delta _5^Ez - {C_E}\gamma _{thc}^E}}}^\infty \!\! {{e^{\! -\! \frac{{\left(\! {\Delta _5^Ez - {C_E}\!\gamma _{thc}^E} \!\right)y \!-\! {M_E}\!\gamma _{thc}^Ez \!-\! \left(\! {{\psi _E} \!+\! \frac{1}{\gamma }} \!\right)\!\gamma _{thc}^E}}{{{\lambda _{SE}}{\xi _E}\gamma _{thc}^E}}}}} }  \\
& \times \frac{1}{{{\lambda _{TE}}}}{e^{ - \frac{y}{{{\lambda _{TE}}}}}}\frac{1}{{{\lambda _{ST}}}}{e^{ - \frac{z}{{{\lambda _{ST}}}}}}dydz.
\end{align}

Similar to the derivation process of $I_{22}$, after some mathematical manipulations, $P_{int}^{T,ni}$ can be obtained.

$\bullet$ Ideal conditions

Substituting $\kappa  = 0$ and $\sigma _e^2 = 0$ into (5), ${C_{{E}}}={M_{{E}}}=0$. Then, the IP of $T$ at the ideal conditions is given by
\begin{align}\nonumber
\label{15}
P_{int}^{T,id} =& \int_0^\infty   \int_{\frac{{\gamma _{thc}^E}}{{\gamma \Delta _5^E}}}^\infty   \left( {1 - {e^{ - \frac{{\Delta _5^E\gamma y - \gamma _{thc}^E}}{}{\lambda _{SE}}{\xi _E}\gamma \gamma _{thc}^E}}} \right)\\
&\frac{2}{{{\lambda _{TE}}{\lambda _{ST}}}}{K_0}\left( {2\sqrt {\frac{y}{{{\lambda _{TE}}{\lambda _{ST}}}}} } \right)dy,
\end{align}

After some mathematical manipulations, we can obtain $P_{int}^{T,id}$.
\bibliographystyle{IEEEtran}
\bibliography{ZML_tvt}

\begin{thebibliography}{10}
\providecommand{\url}[1]{#1}
\csname url@samestyle\endcsname
\providecommand{\newblock}{\relax}
\providecommand{\bibinfo}[2]{#2}
\providecommand{\BIBentrySTDinterwordspacing}{\spaceskip=0pt\relax}
\providecommand{\BIBentryALTinterwordstretchfactor}{4}
\providecommand{\BIBentryALTinterwordspacing}{\spaceskip=\fontdimen2\font plus
\BIBentryALTinterwordstretchfactor\fontdimen3\font minus
  \fontdimen4\font\relax}
\providecommand{\BIBforeignlanguage}[2]{{%
\expandafter\ifx\csname l@#1\endcsname\relax
\typeout{** WARNING: IEEEtran.bst: No hyphenation pattern has been}%
\typeout{** loaded for the language `#1'. Using the pattern for}%
\typeout{** the default language instead.}%
\else
\language=\csname l@#1\endcsname
\fi
#2}}
\providecommand{\BIBdecl}{\relax}
\BIBdecl

\bibitem{1}
X.~{Liu}, H.~{Ding}, and S.~{Hu}, ``{Uplink Resource Allocation for NOMA-based
  Hybrid Spectrum Access in 6G-enabled Cognitive Internet of Things},''
  \emph{IEEE Internet of Things Journal}, pp. 1--1, 2020.

\bibitem{111}
S.~{Jacob}, V.~G. {Menon}, S.~{Joseph}, P.~G. {Vinoj}, A.~{Jolfaei},
  J.~{Lukose}, and G.~{Raja}, ``{A Novel Spectrum Sharing Scheme using Dynamic
  Long Short-Term Memory with CP-OFDMA in 5G Networks},'' \emph{IEEE
  Transactions on Cognitive Communications and Networking}, pp. 1--1, 2020.

\bibitem{2}
Y.~{Liu}, Z.~{Qin}, M.~{Elkashlan}, A.~{Nallanathan}, and J.~A. {McCann},
  ``{Non-Orthogonal Multiple Access in Large-Scale Heterogeneous Networks},''
  \emph{IEEE Journal on Selected Areas in Communications}, vol.~35, no.~12, pp.
  2667--2680, Dec. 2017.

\bibitem{3}
M.~{Zeng}, A.~{Yadav}, O.~A. {Dobre}, G.~I. {Tsiropoulos}, and H.~V. {Poor},
  ``{Capacity Comparison Between MIMO-NOMA and MIMO-OMA With Multiple Users in
  a Cluster},'' \emph{IEEE Journal on Selected Areas in Communications},
  vol.~35, no.~10, pp. 2413--2424, Oct. 2017.

\bibitem{4}
X.~{Li}, J.~{Li}, Y.~{Liu}, Z.~{Ding}, and A.~{Nallanathan}, ``{Residual
  Transceiver Hardware Impairments on Cooperative NOMA Networks},'' \emph{IEEE
  Transactions on Wireless Communications}, vol.~19, no.~1, pp. 680--695, Jan.
  2020.

\bibitem{7527668}
J.~{Choi}, ``{Power Allocation for Max-Sum Rate and Max-Min Rate Proportional
  Fairness in NOMA},'' \emph{IEEE Communications Letters}, vol.~20, no.~10, pp.
  2055--2058, Oct. 2016.

\bibitem{5}
X.~{Lu}, D.~{Niyato}, H.~{Jiang}, D.~I. {Kim}, Y.~{Xiao}, and Z.~{Han},
  ``{Ambient Backscatter Assisted Wireless Powered Communications},''
  \emph{IEEE Wireless Communications}, vol.~25, no.~2, pp. 170--177, Apr. 2018.

\bibitem{6}
B.~{Lyu}, Z.~{Yang}, H.~{Guo}, F.~{Tian}, and G.~{Gui}, ``{Relay Cooperation
  Enhanced Backscatter Communication for Internet-of-Things},'' \emph{IEEE
  Internet of Things Journal}, vol.~6, no.~2, pp. 2860--2871, Apr. 2019.

\bibitem{7}
X.~{Lu}, D.~{Niyato}, H.~{Jiang}, E.~{Hossain}, and P.~{Wang}, ``{Ambient
  Backscatter-Assisted Wireless-Powered Relaying},'' \emph{IEEE Transactions on
  Green Communications and Networking}, vol.~3, no.~4, pp. 1087--1105, Dec.
  2019.

\bibitem{8}
V.~Liu, A.~Parks, V.~Talla, S.~Gollakota, D.~Wetherall, and J.~R. Smith,
  ``{Ambient {B}ackscatter: {W}ireless {C}ommunication out of {T}hin {A}ir},''
  \emph{ACM SIGCOMM}, vol.~43, no.~4, pp. 39--50, Aug. 2013.

\bibitem{9}
D.~{Darsena}, G.~{Gelli}, and F.~{Verde}, ``{Modeling and Performance Analysis
  of Wireless Networks With Ambient Backscatter Devices},'' \emph{IEEE
  Transactions on Communications}, vol.~65, no.~4, pp. 1797--1814, Apr. 2017.

\bibitem{10}
W.~{Zhao}, G.~{Wang}, S.~{Atapattu}, C.~{Tellambura}, and H.~{Guan}, ``{Outage
  Analysis of Ambient Backscatter Communication Systems},'' \emph{IEEE
  Communications Letters}, vol.~22, no.~8, pp. 1736--1739, Aug. 2018.

\bibitem{11}
J.~{Guo}, X.~{Zhou}, S.~{Durrani}, and H.~{Yanikomeroglu}, ``{Design of
  {N}on-{O}rthogonal {M}ultiple {A}ccess {E}nhanced {B}ackscatter
  {C}ommunication},'' \emph{IEEE Trans. Wireless Commun.}, vol.~17, no.~10, pp.
  6837--6852, Oct. 2018.

\bibitem{12}
H.~{Guo}, Y.~{Liang}, R.~{Long}, and Q.~{Zhang}, ``{Cooperative Ambient
  Backscatter System: A Symbiotic Radio Paradigm for Passive IoT},'' \emph{IEEE
  Wireless Communications Letters}, vol.~8, no.~4, pp. 1191--1194, Aug. 2019.

\bibitem{13}
Y.~{Ye}, L.~{Shi}, X.~{Chu}, and G.~{Lu}, ``{On the Outage Performance of
  Ambient Backscatter Communications},'' \emph{IEEE Internet of Things
  Journal}, pp. 1--1, 2020.

\bibitem{14}
X.~{Li}, M.~{Huang}, Y.~{Liu}, V.~G. {Menon}, A.~{Paul}, and Z.~{Ding}, ``{I/Q
  Imbalance Aware Nonlinear Wireless-Powered Relaying of B5G Networks: Security
  and Reliability Analysis},'' \emph{arXiv preprint arXiv:2006.03902,
  \textup{2020. [Online]. Available: https://arxiv.org/abs/2006.03902}}.

\bibitem{144}
M.~{Abbasi}, A.~{Shokrollahi}, M.~R. {Khosravi}, and V.~G. {Menon},
  ``{High-Performance Flow Classification using Hybrid Clusters in Software
  Defined Mobile Edge Computing},'' \emph{Computer Communications,
  \textup{2020. doi: 10.1016/j.comcom.2020.07.002.}}

\bibitem{6772207}
A.~D. {Wyner}, ``{The Wire-tap Channel},'' \emph{The Bell System Technical
  Journal}, vol.~54, no.~8, pp. 1355--1387, Oct. 1975.

\bibitem{9013326}
N.~{Nguyen}, M.~{Zeng}, O.~A. {Dobre}, and H.~V. {Poor}, ``{Securing Massive
  MIMO-NOMA Networks with ZF Beamforming and Artificial Noise},'' in \emph{2019
  IEEE Global Communications Conference (GLOBECOM)}, Feb. 2019, pp. 1--6.

\bibitem{15}
H.~{Lei}, Z.~{Yang}, K.~{Park}, I.~S. {Ansari}, Y.~{Guo}, G.~{Pan}, and
  M.~{Alouini}, ``{Secrecy Outage Analysis for Cooperative NOMA Systems With
  Relay Selection Schemes},'' \emph{IEEE Transactions on Communications},
  vol.~67, no.~9, pp. 6282--6298, Sept. 2019.

\bibitem{16}
B.~{Li}, X.~{Qi}, K.~{Huang}, Z.~{Fei}, F.~{Zhou}, and R.~Q. {Hu},
  ``{Security-Reliability Tradeoff Analysis for Cooperative NOMA in Cognitive
  Radio Networks},'' \emph{IEEE Transactions on Communications}, vol.~67,
  no.~1, pp. 83--96, Jan. 2019.

\bibitem{17}
Q.~{Yang}, H.~{Wang}, Q.~{Yin}, and A.~L. {Swindlehurst}, ``{Exploiting
  Randomized Continuous Wave in Secure Backscatter Communications},''
  \emph{IEEE Internet of Things Journal}, vol.~7, no.~4, pp. 3389--3403, Apr.
  2020.

\bibitem{18}
Y.~{Zhang}, F.~{Gao}, L.~{Fan}, X.~{Lei}, and G.~K. {Karagiannidis}, ``Secure
  {C}ommunications for {M}ulti-{T}ag {B}ackscatter {S}ystems,'' \emph{IEEE
  Wireless Commun. Lett.}, vol.~8, no.~4, pp. 1146--1149, Aug. 2019.

\bibitem{19}
J.~Y. {Han}, M.~J. {Kim}, J.~{Kim}, and S.~M. {Kim}, ``{Physical Layer Security
  in Multi-Tag Ambient Backscatter Communications ¨C Jamming vs.
  Cooperation},'' in \emph{2020 IEEE Wireless Communications and Networking
  Conference (WCNC)}, 2020, pp. 1--6.

\bibitem{8649584}
M.~{Zeng}, N.~{Nguyen}, O.~A. {Dobre}, and H.~V. {Poor}, ``{Securing Downlink
  Massive MIMO-NOMA Networks With Artificial Noise},'' \emph{IEEE Journal of
  Selected Topics in Signal Processing}, vol.~13, no.~3, pp. 685--699, Jun.
  2019.

\bibitem{20}
E.~{Bj$\ddot{\text{o}}$rnson}, J.~{Hoydis}, M.~{Kountouris}, and M.~{Debbah},
  ``{Massive MIMO Systems With Non-Ideal Hardware: Energy Efficiency,
  Estimation, and Capacity Limits},'' \emph{IEEE Transactions on Information
  Theory}, vol.~60, no.~11, pp. 7112--7139, Nov. 2014.

\bibitem{22}
X.~{Li}, M.~{Zhao}, Y.~{Liu}, L.~{Li}, Z.~{Ding}, and A.~{Nallanathan},
  ``{Secrecy Analysis of Ambient Backscatter NOMA Systems under I/Q
  Imbalance},'' \emph{IEEE Transactions on Vehicular Technology}, pp. 1--1,
  2020.

\bibitem{21}
X.~{Li}, Q.~{Wang}, Y.~{Liu}, T.~A. {Tsiftsis}, Z.~{Ding}, and
  A.~{Nallanathan}, ``{UAV-Aided Multi-Way NOMA Networks with Residual Hardware
  Impairments},'' \emph{IEEE Wireless Communications Letters}, pp. 1--1, 2020.

\bibitem{23}
P.~K. {Sharma} and P.~K. {Upadhyay}, ``{Cognitive Relaying With Transceiver
  Hardware Impairments Under Interference Constraints},'' \emph{IEEE
  Communications Letters}, vol.~20, no.~4, pp. 820--823, Apr. 2016.

\bibitem{24}
J.~{Cui}, Z.~{Ding}, and P.~{Fan}, ``{Outage Probability Constrained MIMO-NOMA
  Designs Under Imperfect CSI},'' \emph{IEEE Transactions on Wireless
  Communications}, vol.~17, no.~12, pp. 8239--8255, Dec. 2018.

\bibitem{25}
S.~{Lee}, T.~Q. {Duong}, and R.~{Woods}, ``{Impact of Wireless Backhaul
  Unreliability and Imperfect Channel Estimation on Opportunistic NOMA},''
  \emph{IEEE Transactions on Vehicular Technology}, vol.~68, no.~11, pp.
  10\,822--10\,833, Nov. 2019.

\bibitem{26}
J.~{He}, Z.~{Tang}, Z.~{Tang}, H.~{Chen}, and C.~{Ling}, ``{Design and
  Optimization of Scheduling and Non-Orthogonal Multiple Access Algorithms With
  Imperfect Channel State Information},'' \emph{IEEE Transactions on Vehicular
  Technology}, vol.~67, no.~11, pp. 10\,800--10\,814, Nov. 2018.

\bibitem{27}
A.~K. {Mishra}, D.~{Mallick}, and P.~{Singh}, ``{Combined Effect of RF
  Impairment and CEE on the Performance of Dual-Hop Fixed-Gain AF Relaying},''
  \emph{IEEE Communications Letters}, vol.~20, no.~9, pp. 1725--1728, Sept.
  2016.

\bibitem{8879726}
X.~{Li}, M.~{Huang}, J.~{Li}, Q.~{Yu}, K.~{Rabie}, and C.~C. {Cavalcante},
  ``{Secure Analysis of Multi-antenna Cooperative Networks with Residual
  Transceiver HIs and CEEs},'' \emph{IET Communications}, vol.~13, no.~17, pp.
  2649--2659, Oct. 2019.

\bibitem{7546865}
X.~{Ding}, T.~{Song}, Y.~{Zou}, X.~{Chen}, and L.~{Hanzo},
  ``{Security-Reliability Tradeoff Analysis of Artificial Noise Aided Two-Way
  Opportunistic Relay Selection},'' \emph{IEEE Transactions on Vehicular
  Technology}, vol.~66, no.~5, pp. 3930--3941, May 2017.

\bibitem{28}
{Taesang Yoo} and A.~{Goldsmith}, ``{Capacity and Power Allocation for Fading
  MIMO Channels with Channel Estimation Error},'' \emph{IEEE Transactions on
  Information Theory}, vol.~52, no.~5, pp. 2203--2214, May 2006.

\bibitem{29}
S.~Stefania, B.~Matthew, and T.~Issam, \emph{{{{LTE}-the {UMTS} Long Term
  Evolution: From Theory to Practice}}}, 2nd~ed.\hskip 1em plus 0.5em minus
  0.4em\relax New York, NY, USA: Wiley \& Sons, 2011.

\bibitem{31}
{M. Abramowitz and I. A. Stegun}, \emph{{{Handbook of Mathematical Functions
  with Formulas, Graphs, and Mathematical Tables, 10th ed.}}}\hskip 1em plus
  0.5em minus 0.4em\relax {New York, NY, USA}: {Academic}, 1972.

\bibitem{30}
{Biglieri, Ezio and Calderbank, Robert and Constantinides, Anthony and
  Goldsmith, Andrea and Paulraj, Arogyaswami and Poor, H Vincent},
  \emph{{{MIMO} Wireless Communications}}.\hskip 1em plus 0.5em minus
  0.4em\relax Cambridge university press, 2007.

\bibitem{9032127}
X.~{Li}, M.~{Zhao}, X.~{Gao}, L.~{Li}, D.~{Do}, K.~M. {Rabie}, and R.~{Kharel},
  ``{Physical Layer Security of Cooperative NOMA for IoT Networks Under I/Q
  Imbalance},'' \emph{IEEE Access}, vol.~8, pp. 51\,189--51\,199, Mar. 2020.

\bibitem{32}
{I. S. Gradshteyn and I. M. Ryzhik}, \emph{{Table of {I}ntegrals, {S}eries, and
  {P}roducts}}.\hskip 1em plus 0.5em minus 0.4em\relax {New York, NY, USA}:
  {Academic Press}, 2007.

\bibitem{33}
{F. B. Hildebrand}, \emph{{{Introduction to Numerical Analysis}}}.\hskip 1em
  plus 0.5em minus 0.4em\relax {New York, USA}: {Dover Publications}, 1987.

\end{thebibliography}

\end{document}